\documentclass[oneside,english,10pt]{amsart}
\usepackage{color}
\usepackage{amsthm}
\usepackage{amstext}
\usepackage{graphicx}
\usepackage{amssymb}
\usepackage{esint}
\usepackage[all]{xy}
\usepackage{hyperref}

\makeatletter
\numberwithin{equation}{section} 
\numberwithin{figure}{section} 
\theoremstyle{plain}
\newtheorem{thm}{Theorem}[subsection]

\newtheorem{prop}[thm]{Proposition}

\newtheorem{cor}[thm]{Corollary}

\theoremstyle{definition}

\newtheorem{expl}[thm]{Example}

\newtheorem{rem}[thm]{Remark}

\newenvironment{nouppercase}{%
  \renewcommand{\uppercasenonmath}[1]{}}{}

\usepackage{babel}

\begin{document}
\begin{nouppercase}

\title{Differential Rings from Special
K\"ahler Geometry}
\author{Jie Zhou }
\address{Department of Mathematics, Harvard University, One Oxford Street, Cambridge, MA 02138, USA}
\email{jiezhou@math.harvard.edu}


\begin{abstract}
We study triples of graded rings defined over the deformation
spaces for certain one-parameter families of Calabi-Yau threefolds.
These rings are analogues of the rings of modular forms, quasi-modular forms and almost-holomorphic modular forms. We also discuss
some of their applications in solving the holomorphic anomaly
equations.
\end{abstract}
\maketitle


\setcounter{tocdepth}{2}{\small \tableofcontents{}}{\small \par}

\section{\textup{Introduction}}

Modular forms are very interesting objects in mathematics and
physics due to their nice behaviors under the transformations in the
corresponding modular groups. Given a genus zero subgroup $\Gamma$
of finite index of the full modular group
$\Gamma(1)=\textrm{PSL}(2,\mathbb{Z})$, the generators of the ring
of modular forms for $\Gamma$ could be obtained starting from some
$\theta$ or $\eta$ functions or Eisenstein series. Alternatively,
one could parameterize the modular curve
$X_{\Gamma}=\Gamma\backslash \mathcal{H}^{*}$ by a Hauptmodul
$\alpha$ (generator of the rational functional field of the modular
curve), and consider the periods which are solutions to the
Picard-Fuchs equation attached to the elliptic curve family
$\pi_{\Gamma}: \mathcal{E}_{\Gamma}\rightarrow X_{\Gamma}$ 
constructed from the modular group $\Gamma$. Knowing the relation
between the Hauptmodul $\alpha$ and the normalized period $\tau$ of
the elliptic curve family, one could then obtain the graded
differential ring of quasi-modular forms $\widetilde{M}_{*}(\Gamma)$
by taking successive derivatives of the periods with respect to
$\tau$. This ring $\widetilde{M}_{*}(\Gamma)$ includes the ring of
modular forms $M_{*}(\Gamma)$ and contains further elements which
are not modular but quasi-modular \cite{Kaneko:1995}. The quasi-modular forms in $\widetilde{M}_{*}(\Gamma)$ could be completed to modular forms by
adding some non-holomorphic parts. Then
one gets the graded differential ring of almost-holomorphic modular
forms $\widehat{M}_{*}(\Gamma)$. See \cite{Kaneko:1995, Zagier:2008}
and references therein for details about the construction.

Similarly, as we shall show in this work, for certain one-parameter Calabi-Yau threefold family
$\pi:\mathcal{X}\rightarrow \mathcal{M}$, we can define a triple of
graded rings
$(\mathcal{R},\widetilde{\mathcal{R}},\widehat{\mathcal{R}})$. The
analogue of the modular variable $\tau$ is constructed using the
special K\"ahler geometry \cite{Strominger:1990pd} on the base
$\mathcal{M}$, and the periods are solved from the Picard-Fuchs
equation of the Calabi-Yau family. This triple share similar
structures and operations with the triple
$(M_{*}(\Gamma),\widetilde{M}_{*}(\Gamma), \widehat{M}(\Gamma))$
defined for the elliptic curve family $\pi_{\Gamma}:
\mathcal{E}_{\Gamma}\rightarrow X_{\Gamma}$. In particular, they have
gradings which play the role of modular weights.\\

These rings appeared in various forms in the studies of BCOV
holomorphic anomaly equations \cite{Bershadsky:1993ta,
Bershadsky:1993cx} and topological string partition functions for
Calabi-Yau threefold families \cite{Bershadsky:1993cx,
Yamaguchi:2004bt, Alim:2007qj, Alim:2013} (see also
\cite{Huang:2006si, Aganagic:2006wq, Grimm:2007tm, Hosono:2008ve,
Movasati:2011zz} for related works). In particular, the ring
$\widehat{\mathcal{R}}$ is the ring constructed by Yamaguchi-Yau in
\cite{Yamaguchi:2004bt} (see also \cite{Alim:2007qj}). In the recent
work \cite{Alim:2013}, it was shown that for some special
one-parameter Calabi-Yau threefold families, the normalized
topological string partition functions are polynomials in the
generators of the ring $\widehat{\mathcal{R}}$ and have weight zero.

For some particular one-parameter families of non-compact Calabi-Yau
threefolds with an elliptic curve sitting inside each fiber
\cite{Lian:1994zv, Klemm:1996, Lerche:1996ni, Chiang:1999tz,
Hori:2000kt}, the base $\mathcal{M}$ of the Calabi-Yau threefold
family $\pi:\mathcal{X}\rightarrow \mathcal{M}$ could be identified
with a modular curve $X_{\Gamma}$ \cite{Alim:2013} (see also
\cite{Klemm:1999gm, Aganagic:2006wq, Haghighat:2008gw}). The triple
of rings
$(\mathcal{R},\widetilde{\mathcal{R}},\widehat{\mathcal{R}})$ we
shall define are closely related to the known graded rings
$(M_{*}(\Gamma),\widetilde{M}_{*}(\Gamma), \widehat{M}(\Gamma))$.
This then allows one to express the coefficients and derivatives in
the holomorphic anomaly equations purely in the language of modular
form theory. Using the polynomial recursion technique \cite{Bershadsky:1993cx, Yamaguchi:2004bt, Alim:2007qj}, one can then
prove \cite{Alim:2013} by induction that the topological partition
functions, as solutions to the holomorphic anomaly equations, are
almost-holomorphic modular forms of weight $0$. Moreover, for each
of these particular non-compact Calabi-Yau threefold families, the
boundary conditions, used in e.g., \cite{Bershadsky:1993cx,
Bershadsky:1993ta, Huang:2006hq, Huang:2006si} to solve the
topological string partition functions, become regularity conditions
\cite{Alim:2013} at different cusps on the modular curve for the
holomorphic limits of the topological string partition functions
which are quasi-modular forms. Then, by using the Fricke involution
(also called Atkin--Lehner involution) which is an automorphism of
the corresponding modular curve, one could translate the regularity
conditions at different cusps to some conditions on the
$q_{\tau}=\exp 2\pi i \tau$ expansions of the quasi-modular forms.
In this way, solving the holomorphic anomaly equations becomes a
purely mathematical problem in which one needs to solve for certain
quasi-modular forms from the recursive equations and the regularity
conditions. In \cite{Alim:2013}, the first few topological string
partition functions for these special Calabi-Yau threefold families
were obtained genus by genus, and were expressed in terms of
generators of the ring of almost-holomorphic
modular forms.\\

Experiences from dealing with the non-compact Calabi-Yau threefolds
suggest that studying the arithmetic properties of the rings
$\mathcal{R},\widetilde{\mathcal{R}},\widehat{\mathcal{R}}$ might be
useful in solving the holomorphic
anomaly equations for more general Calabi-Yau families. This is the motivation of this work.\\

This paper is based on \cite{Alim:2013} and follows closely the
lines of thoughts in \cite{Ceresole:1992su, Ceresole:1993qq, Hosono:2008ve}. In particular,  in \cite{Hosono:2008ve}
some differential rings were constructed using special K\"ahler
geometry, and the parallelism to the quasi-modular forms and
almost-holomorphic modular forms was made. The main difference
between our present work and \cite{Hosono:2008ve} is that in the
latter work the rings were constructed from the rings of propagators
\cite{Bershadsky:1993cx, Yamaguchi:2004bt, Alim:2007qj}, while we
construct the rings from the periods and Yukawa couplings solved
from the Picard-Fuchs equations in such a way that the analogue
between these rings and the rings of modular objects is more
explicit. Many of the discussions in the current work are inspired
by the examples discussed in \cite{Hosono:2008ve}.

The structure of this paper is as follows. In section 2, we shall review
the preliminaries about modular curves, modular forms, Picard-Fuchs
equations, special K\"ahler geometry and BCOV holomorphic anomaly
equations, and set up the notations. In section 3, first we
reproduce the details in constructing the graded rings
$M_{*}(\Gamma),\widetilde{M}_{*}(\Gamma),\widehat{M}_{*}(\Gamma)$
for certain elliptic curve families. Then we construct by analogue
the graded rings $ \mathcal{R}, \widetilde{\mathcal{R}}$ for certain
one-parameter Calabi-Yau threefold families
$\pi:\mathcal{X}\rightarrow \mathcal{M}$. In section 4, we consider
these rings using the properties of the special K\"ahler geometry on
the deformation space $\mathcal{M}$. We make use of the canonical
coordinates and holomorphic limits to lift the ring
$\widetilde{\mathcal{R}}$ to a non-holomorphic ring
$\widehat{\mathcal{R}}$, by relating them to the ring constructed by
Yamaguchi-Yau \cite{Yamaguchi:2004bt}. The \textbf{main results} of
this paper are summarized in section 4.5. Some of their applications
in solving the BCOV holomorphic anomaly equations are also discussed
in section 4. We conclude with some discussions and questions in
section 5.


\medskip{}

\textbf{Acknowledgements:} The author would like to thank his
advisor Shing-Tung Yau for many inspiring discussions and constant
support and encouragements. He also thanks Murad Alim, Emanuel
Scheidegger and Shing-Tung Yau for valuable collaborations and
discussions on related projects. He thanks Murad Alim in particular
for teaching him the polynomial recursion technique in solving the
BCOV holomorphic anomaly equations, and many other topics in string
theory. Thanks also go to Chen-Yu Chi, Yaim Cooper, Teng Fei,
Shinobu Hosono, Si Li, Hossein Movasati, Jan Stienstra and Don Zagier for helpful
discussions and correspondences. Part of the work is done while the
author was visiting the Department of Mathematics in National Taiwan
University and the Fields Institute in Toronto, he thanks them for
hospitality and financial support.

\section{\textup{Preliminaries and background}}

To make the paper self-contained, in this section we give a review
of some basic facts about modular groups and modular curves, modular
forms, Picard-Fuchs equations, special K\"ahler geometry and BCOV
holomorphic anomaly equations.

\subsection{Modular groups and modular curves}

The generators and relations for the group
$\textrm{SL}(2,\mathbb{Z})$ are given by the following:
\begin{equation*}
 T =
  \begin{pmatrix}
    1 & 1\\ 0 & 1
  \end{pmatrix}\,,\quad
  S=
  \begin{pmatrix}
    0& -1 \\ 1 & 0\\
  \end{pmatrix}\,,\quad S^{2}=-I\,,\quad (ST)^{3}=-I\,.
\end{equation*}
The groups we are interested in in this paper are the Hecke
subgroups of
$\Gamma(1)=\textrm{PSL}(2,\mathbb{Z})=\textrm{SL}(2,\mathbb{Z})
/\{\pm I\}$:
\begin{equation}
\Gamma_{0}(N)=\left\{ \left.
\begin{pmatrix}
a & b  \\
c & d
\end{pmatrix}
\right\vert\, c\equiv 0\,~ \textrm{mod} \,~ N\right\}< \Gamma(1)\,,
\end{equation}
with $N=2,3,4$. A further subgroup that we will consider is the
unique normal subgroup in $\Gamma(1)$ of index 2. it is generated by
$T^{2},T^{-1}S$ and is often denoted by $\Gamma_0(1)^*$. By abuse of
notation, we write it as $\Gamma_{0}(N)$ with $N=1^*$ when listing
it together with the groups $\Gamma_0(N)$.

The group $\textrm{SL}(2,\mathbb{Z})$ acts on the upper half plane
$\mathcal{H} = \{ \tau \in \mathcal{H} |\, \text{Im} \tau > 0 \}$ by
fractional linear transformations:
\begin{equation}
\tau \mapsto \gamma\tau=\frac{a\tau+b}{c\tau+d}\quad \text{for}
\quad \gamma=\begin{pmatrix} a&b\\c&d\end{pmatrix} \in
\textrm{SL}(2,\mathbb{Z})\,.
\end{equation}
The quotient space $Y_{0}(N)= \Gamma_{0}(N)\backslash \mathcal{H}$
is a non-compact orbifold with certain punctures corresponding to
the cusps and orbifold points corresponding to the elliptic points
of the group $\Gamma_{0}(N)$. By filling the punctures, one then
gets a compact orbifold
$X_{0}(N)=\overline{Y_{0}(N)}=\Gamma_{0}(N)\backslash
\mathcal{H}^{*}$ where $\mathcal{H}^* = \mathcal{H} \cup \{i\infty\}
\cup \mathbb{Q}$. The orbifold $X_0(N)$ can be equipped with the
structure of a Riemann surface. The signature for the group
$\Gamma_{0}(N)$ and the two orbifolds $Y_{0}(N),X_{0}(N)$ could be
represented by $\{p,\mu;\nu_{2},\nu_{3},\nu_{\infty}\}$, where $p$
is the genus of $X_{0}(N)$, $\mu$ is the index of $\Gamma_{0}(N)$ in
$\Gamma(1)$, and $\nu_{i}$ are the numbers of
$\Gamma_{0}(N)$-equivalent elliptic fixed points or parabolic fixed
points of order $i$. The signatures for the groups $\Gamma_{0}(N)$,
$N=1^*,2,3,4$ are listed in the following table (see
e.g.~\cite{Rankin:1977ab}):
\begin{equation*}\label{signature}
  \begin{array}[h]{|c|c|c|c|c|c|}
    \hline
    N & \nu_2 & \nu_3 & \nu_\infty & \mu & p\\
    \hline
    1^* & 0 & 1 & 2 & 2 & 0 \\
    2 & 1 & 0 & 2 & 3 & 0 \\
    3 & 0 & 1 & 2 & 4 & 0\\
    4 & 0 & 0 & 3 & 6 & 0\\
   \hline
  \end{array}
\end{equation*}
The space $X_{0}(N)$ is called a modular curve. When
$N=1^{*},2,3,4$, it has three singular points corresponding to the
above $\Gamma_{0}(N)$--equivalent fixed points. There is an action
called Fricke involution
\begin{equation}\label{Fricke}
W_{N}: \tau\mapsto -{1\over N\tau}
\end{equation}
on the modular curve $X_{0}(N)$. It exchanges the two cusp
classes\footnote{We use the notation $[\tau]$ to denote the
equivalence class of $\tau\in \mathcal{H}^{*}$ under the group
action of $\Gamma$ on $\mathcal{H}^{*}$.} $[i\infty]$ and $[0]$,
while fixing the fixed point.

More details about the basic theory could be found in e.g.
\cite{Rankin:1977ab,Diamond:2005,Zagier:2008}.

\subsection{Modular forms}

We proceed by recalling some basic concepts in modular form theory
following~\cite{Kaneko:1995} (see also \cite{Zagier:2008} for more
details). In the following, we shall use the notation $\Gamma$ for a
general subgroup of finite index in $\Gamma(1)$. In particular, we
can take $\Gamma$ to be the modular groups $\Gamma_{0}(N)$ described
above.

A modular form of weight $k$ with respect to the group $\Gamma$ is
defined to be a holomorphic function $f$ on $\mathcal{H}$ satisfying
\begin{equation}
f({a\tau+b\over c\tau+d} )=(c\tau+d)^{k}f(\tau),\,\forall~ \tau\in
\mathcal{H},\, \begin{pmatrix}
a & b  \\
c & d
\end{pmatrix}\in\Gamma\,.
\end{equation}
and growing at most polynomially in ${1\over \mathrm{Im}\tau}$ as
${1\over \mathrm{Im}\tau}\rightarrow 0$. One can also define modular
forms with characters, see for example \cite{Diamond:2005} for
details. The space of holomorphic modular forms for $\Gamma$ forms a
graded ring and is denote by $M_{*}(\Gamma)$.

A quasi-modular form of weight $k$ with respect to the group
$\Gamma$ is a holomorphic function satisfying the same growth
conditions but with the transformation properties replace by: there
exist holomorphic functions $f_{i}, i=0,1,2,3,\dots, k-1$ such that
\begin{equation}
f( {a\tau+b\over c\tau+d}
)=(c\tau+d)^{k}f(\tau)+\sum_{i=0}^{k-1}c^{i}(c\tau+d)^{k-i}f_{i}(\tau)\,,\quad
\forall \tau\in \mathcal{H},\, \begin{pmatrix}
a & b  \\
c & d
\end{pmatrix}\in \Gamma\,.
\end{equation}
The space of quasi-modular forms for $\Gamma$ forms a graded
differential ring and is denote by $\widetilde{M}_{*}(\Gamma)$.

An almost-holomorphic modular form is a function
$f(\tau,\bar{\tau})$ on $\mathcal{H}$ which satisfies the same
growth conditions as above and transforms as a modular form of
weight $k$:
\begin{equation}
f( {a\tau+b\over c\tau+d},\overline{{a\tau+b\over c\tau+d}}
)=(c\tau+d)^{k}f(\tau,\bar{\tau}),\quad \forall~ \tau\in
\mathcal{H},\quad
\begin{pmatrix}
a & b  \\
c & d
\end{pmatrix}\in \Gamma\,.
\end{equation}

It can be written in the form \cite{Kaneko:1995}
\begin{equation}\label{decompostionofalmosthol}
f(\tau,\bar{\tau})=\sum_{i=0}^{m=[k/2]}f_{m}(\tau)Y^{m},\, Y={1\over
\mathrm{Im}\tau}\,
\end{equation}
where $f_{m}(\tau)$ are holomorphic functions on $\mathcal{H}$ for
$m=0,1,2,\cdots [k/2]$. The space of almost-holomorphic modular
forms for $\Gamma$ forms a graded differential ring and is denote by
$\widehat{M}_{*}(\Gamma)$.

As shown in \cite{Kaneko:1995}, one has the ring isomorphism
$\widetilde{M}_{*}(\Gamma)=M_{*}(\Gamma)\otimes \mathbb{C}[E_{2}]$,
where $E_{2}$ is the Eisenstein series defined by
\begin{equation*}
E_{2}(\tau)=1-24\sum_{k=1}^{\infty}\sigma_{1}(k)q^{k},\, q=e^{2 \pi
i\tau},\, \sigma_{1}(k)=\sum_{d:d|k}d.
\end{equation*}
Moreover, there is a ring isomorphism
$\widehat{M}_{*}(\Gamma)\rightarrow \widetilde{M}_{*}(\Gamma)$
defined by $f(\tau,\bar{\tau})\mapsto f_{0}(\tau)$, where $f_{0}$ is
the function in (\ref{decompostionofalmosthol}). If one regards $Y$
as a formal variable, then this is the ``constant term map" obtained
by taking the limit $ Y={1\over \mathrm{Im}\tau}\rightarrow 0$
(which could be induced from $\bar{\tau}\rightarrow \infty$, by
thinking of $\bar{\tau}$ as a complex coordinate independent of
$\tau$). The inverse map takes a quasi-modular form to an
almost-holomorphic modular form, we shall call this map ``modular
completion" in this paper.

\begin{expl}
Take the group $\Gamma$ to be the full modular group
$\Gamma(1)=\textrm{PSL}(2,\mathbb{Z})$. Then
$M_{*}(\Gamma(1))=\mathbb{C}[E_{4},E_{6}]$, where $E_{4},E_{6}$ are
the familiar Eisenstein series defined by
\begin{eqnarray*}
E_{4}(\tau)&=&1+240\sum_{k=1}^{\infty}\sigma_{3}(k)q^{k},\, q=e^{2
\pi
i\tau},\, \sigma_{3}(k)=\sum_{d:d|k}d^{3}\,,\\
E_{6}(\tau)&=&1-504\sum_{k=1}^{\infty}\sigma_{3}(k)q^{k},\, q=e^{2
\pi i\tau},\, \sigma_{5}(k)=\sum_{d:d|k}d^{5}\,.
\end{eqnarray*}
The Eisenstein series $E_{2}$ is a quasi-modular form for
$\Gamma(1)$ since it transforms according to
\begin{equation*}
E_{2}({a\tau+b\over c\tau+d})=(c\tau+d)^{2}E_{2}(\tau)+{12\over 2\pi
i}c(c\tau+d),\quad \forall ~\tau\in \mathcal{H},\quad \forall
\begin{pmatrix}
a & b  \\
c & d
\end{pmatrix} \in \Gamma(1)\,.
\end{equation*}
Recall that
\begin{equation*}
{1\over \textrm{Im}{a\tau+b\over c\tau+d}}=(c\tau+d)^{2}{1\over
\mathrm{Im}\tau}-2ic(c\tau+d),\quad \forall~ \tau\in
\mathcal{H},\quad \forall
\begin{pmatrix}
a & b  \\
c & d
\end{pmatrix} \in \Gamma(1)\,.
\end{equation*}
we know the modular completion of the quasi-modular form
$E_{2}(\tau)$ is
\begin{equation*}
\hat{E_{2}}(\tau,\bar{\tau})=E_{2}(\tau)-{3\over
\pi\mathrm{Im}\tau}\,.
\end{equation*}
Then $M_{*}(\Gamma(1))=\mathbb{C}[E_{4},E_{6}],
\widetilde{M}_{*}(\Gamma(1))=\mathbb{C}[E_{2},E_{4},E_{6}],
\widehat{M}_{*}(\Gamma(1))=\mathbb{C}[\hat{E_{2}},E_{4},E_{6}]$. The
latter two carry differential structures given by
\begin{eqnarray*}
DE_{2}&=&{1\over 12}(E_{2}^{2}-E_{4}),\, DE_{4}={1\over
3}(E_{2}E_{4}-E_{6}),\,  DE_{6}={1\over
2}(E_{2}E_{6}-E_{4}^{2})\,,\\
\hat{D}\hat{E}_{2}&=&{1\over 12}(\hat{E}_{2}^{2}-E_{4}),\,
\hat{D}E_{4}={1\over 3}(\hat{E}_{2}E_{4}-E_{6}),\,
\hat{D}E_{6}={1\over 2}(\hat{E}_{2}E_{6}-E_{4}^{2})\,,
\end{eqnarray*}
where $D={1\over 2\pi i}{\partial\over \partial \tau}:
\widetilde{M}_{k}(\Gamma(1))\rightarrow
\widetilde{M}_{k+2}(\Gamma(1))$ and $\hat{D}=D+{k\over 12}\cdot
{-3\over \pi\mathrm{Im}\tau}:\widehat{M}_{k}(\Gamma(1))\rightarrow
\widehat{M}_{k+2}(\Gamma(1))$.
\end{expl}

\subsection{Picard-Fuchs equations for families of Calabi-Yau manifolds }

Throughout this work, we do not aim to give a complete discussion of
all Picard-Fuchs equations for Calabi-Yau families. Instead, the
examples we shall study are the Picard-Fuchs equations for some
special Calabi-Yau families given in (\ref{PFforonefold}),
(\ref{PFfornoncompactCY}) and (\ref{PFforcompactCY}) below. These
equations are considered in e.g., \cite{Lian:1994zv, Lian:1995,
Lian:1996, Klemm:1996, Lerche:1996ni, Chiang:1999tz} in which the
arithmetic properties like the integrality of the mirror maps and
modular relations are studied.  More general Picard-Fuchs equations
attached to Calabi-Yau families are considered in e.g.,
\cite{Lian:1995, Lian:1996, Verrill:2000, van:2006, Yang:2007,
Chen:2008}.

Consider a family $\pi: \mathcal{X}\rightarrow \mathcal{M}$ of
Calabi-Yau $n$--folds $\mathcal{X}=\{\mathcal{X}_{z}\}$ over a
variety $\mathcal{M}$ parametrized by the complex coordinate system
$z=\{z^{i}\}_{i=1,2,\cdots \dim \mathcal{M}}$. For a generic $z\in
\mathcal{M}$, the fiber $\mathcal{X}_{z}$ is a smooth Calabi-Yau
$n$--fold. We also assume that $\dim
\mathcal{M}=h^{1}(\mathcal{X}_{z},T\mathcal{X}_{z})$ for a smooth
$\mathcal{X}_{z}$, where $T\mathcal{X}_{z}$ is the holomorphic
tangent bundle of $\mathcal{X}_{z}$. The periods are defined to be
the integrals $\int_{C}\Omega_{z}$, where $C\in
H_{n}(\mathcal{X}_{z},\mathbb{Z})$ and $\Omega=\{\Omega_{z}\}$ is a
holomorphic section of the Hodge line bundle
$\mathcal{L}=\mathcal{R}^{0}\pi_{*}\Omega^{n}_{\mathcal{M}|{\mathcal{X}}}$
on $\mathcal{M}$. They satisfy a differential equation called the
Picard-Fuchs equation induced from the Gauss-Manin connection on the
Hodge bundle
$\mathcal{H}=\mathcal{R}^{n}\pi_{*}\underline{\mathbb{C}}\otimes
\mathcal{O}_{\mathcal{M}}$. In the following, we shall use the
notation $X$ to denote a generic fiber $\mathcal{X}_{z}$ in the
family without specifying the point $z$ and use $\Pi$ to denote a
period. We will also call the base $\mathcal{M}$ the
deformation space (of complex structures of $X$).\\

The families of Calabi-Yau one--folds that we shall focus on in this
work are the elliptic curve families $\pi_{\Gamma_{0}(N)}:
\mathcal{E}_{\Gamma_{0}(N)}\rightarrow
X_{0}(N)=\Gamma_{0}(N)\backslash \mathcal{H}^{*}$ with
$N=1^{*},2,3,4$, where $\mathcal{E}_{\Gamma_{0}(N)}$ is the surface \cite{Kodaira:1963, Shioda:1972}
defined by
\begin{eqnarray}
&\mathcal{E}_{\Gamma_{0}(N)}&:=(\Gamma_{0}(N)\rtimes \mathbb{Z}^{2})\backslash
\left(\mathcal{H}^{*}\times \mathbb{C}\right)\,, \\
&\left( \gamma=\begin{pmatrix}
a & b  \\
c & d
\end{pmatrix},(m,n)\right)&:(\tau,z)\mapsto \left({a\tau+b\over a\tau+d},{z+m\tau+n\over
c\tau+d}\right), \quad\forall ~\gamma\in \Gamma_{0}(N)\,.\nonumber
\end{eqnarray}
The explicit equations, $j$--invariants and Picard-Fuchs operators
of these families could be found in e.g., \cite{Lian:1994zv,
Klemm:1996}. In the following we shall only display the Picard-Fuchs
operators
\begin{equation}\label{PFforonefold}
\mathcal{L}_{\textrm{elliptic}}=\theta^{2}-\alpha
(\theta+1/r)(\theta+1-1/r),\quad \theta=\alpha{\partial\over
\partial \alpha}\,,
\end{equation}
where $r=6,4,3,2$ for $N=1^{*},2,3,4$, respectively. The parameter
$\alpha$ is the complex coordinate on the deformation space
$\mathcal{M}$ in which the Picard-Fuchs equation takes the above
particular form. Thinking of the base space $\mathcal{M}$ as the
genus zero modular curve $X_{0}(N)$, it is then a modular function
(called Hauptmodul) for the modular group $\Gamma_{0}(N)$. Each of
these Picard-Fuchs equations has three regular singularities located
at $\alpha=0,1,\infty$ on the corresponding modular curve. The two
points $\alpha=0,1$ are the cusp classes $[i\infty],[0]$
respectively, and are exchanged by the
Fricke involution $W_{N}$ given in (\ref{Fricke}).\\

We shall also consider the Picard-Fuchs equations for the mirror
families of the $K_{\mathbb{P}^{2}}$ and $K_{dP_{n}},n=5,6,7,8$,
where $dP_{n}$ is the del Pezzo surface obtained from blowing up
$\mathbb{P}^{2}$ at $n$ points. We take the Kahler structures of the
non-compact Calabi-Yau threefolds $K_{\mathbb{P}^{2}}, K_{dP_{n}},
n=5,6,7,8$ to be the ones induced from the hyperplane classes on the
$\mathbb{P}^{2}$. The explicit equations and Picard-Fuchs equations
of these families are given in \cite{Lerche:1996ni, Chiang:1999tz}.
The Picard-Fuchs equations have the following form:
\begin{equation}\label{PFfornoncompactCY}
\mathcal{L}_{\textrm{CY}}=\mathcal{L}_{\textrm{elliptic}}\circ
\theta=\left(\theta^{2}-\alpha
(\theta+1/r)(\theta+1-1/r)\right)\circ \theta\,,
\end{equation}
where $\mathcal{L}_{\textrm{elliptic}}$ and $\theta$ are the same as
those described in (\ref{PFforonefold}).

To make things more precise, we shall discuss the mirror Calabi-Yau
threefold family $\pi:\mathcal{X}\rightarrow \mathcal{M}$ of the
$K_{\mathbb{P}^{2}}$ family, constructed in \cite{Chiang:1999tz,
Hori:2000kt}. We refer the interested readers to
\cite{Lerche:1996ni, Alim:2013} and references therein for the
detailed discussion on other families. For each $\alpha$ on
$\mathcal{M}$, the fiber $\mathcal{X}_{\alpha}$ of the non-compact
Calabi-Yau threefold family is a conic fibration given by
\begin{equation*}
\mathcal{X}_{\alpha}: uv-H(x,y;\alpha)=0,\quad (u,v,x,y)\in
\mathbb{C}^{2}\times (\mathbb{C}^{*})^{2}\,,
\end{equation*}
where $H(x,y,\alpha)=y^{2}-(x+1)y-\alpha x^{3}$. The degeneration
locus of this conic fibration is the elliptic curve
$\mathcal{E}_{\alpha}: H(x,y;\alpha)=0$. Then as $\alpha$ moves in
$\mathcal{M}$, one gets an elliptic curve family
$\pi:\mathcal{E}\rightarrow \mathcal{M}$. It turns out that this
elliptic curve family is 3-isogenous to the Hesse cubic curve family
$\pi_{\Gamma_{0}(3)}: \mathcal{E}_{\Gamma_{0}(3)}\rightarrow
X_{0}(3)$, see e.g., \cite{Husemoller:2004} for details. Then we see
$\mathcal{M}\cong X_{0}(3)$. Moreover, the Picard-Fuchs equation for
the Calabi-Yau threefold family is \cite{Lerche:1996ni,
Chiang:1999tz}
\begin{equation*}
\mathcal{L}_{\textrm{CY}}=\mathcal{L}_{\textrm{elliptic}}\circ
\theta=\left(\theta^{2}-\alpha (\theta+1/3)(\theta+2/3)\right)\circ
\theta\,.
\end{equation*}
Near the point $\alpha=0$, one can find a basis of solutions to
$\mathcal{L}_{\textrm{CY}}$ given by
\begin{equation*}
X^{0}=1,\quad t\sim \log\alpha+\mathcal{O}(\alpha^{0}),\quad
F_{t}\,,
\end{equation*}
so that $\theta t,\frac{1}{2\pi i}\kappa^{-1}\theta F_{t}$ are the
periods $\omega_{0},\omega_{1}=\tau \omega_{0}$ of the elliptic
curve $\mathcal{E}_{\alpha}$, respectively. Then
\begin{equation*}\label{tauFtt}
\tau={\omega_{1}\over\omega_{0}}={{1\over 2\pi i}\kappa^{-1}\theta
F_{t}\over \theta t}={1\over 2\pi i}\kappa^{-1}F_{tt}\,.
\end{equation*}
where $\kappa=-{1\over 3}$ is the classical triple intersection of
$K_{\mathbb{P}^{2}}$.
 This quantity will prove to be the key of understanding the
arithmetic
properties of special K\"ahler geometry later.\\

The compact Calabi-Yau threefold family that we shall study is the
quintic mirror family \cite{Candelas:1990rm} whose Picard-Fuchs
equation is given by
\begin{equation}\label{PFforcompactCY}
\mathcal{L}=\theta^{4}-\alpha (\theta+1/5) (\theta+2/5) (\theta+3/5)
(\theta+4/5)\,,
\end{equation}
where $\alpha$ is related to the parameter $z$ in
\cite{Candelas:1990rm} by $\alpha=5^{5}z$.\\

For the Picard-Fuchs equations in (\ref{PFfornoncompactCY}),
(\ref{PFforcompactCY}), they all have three regular singularities
located at $\alpha=0,1,\infty$ on the base $\mathcal{M}$. The point
$\alpha=0$ on the deformation space $\mathcal{M}$ plays a special
role and is called the large complex structure limit of the
Calabi-Yau threefold family. The solutions to these equations could
be obtained using the Frobenius method and are related to
hypergeometric functions.

\subsection{Special K\"ahler geometry on deformation spaces}

The discussions on special K\"ahler geometry and holomorphic anomaly
equations in this section apply to multi-parameter Calabi-Yau
families.\\

Consider a family of Calabi-Yau threefolds $X$ given by
$\pi:\mathcal{X}\rightarrow \mathcal{M}$ as above. The base (called
deformation space above) is equipped with the Weil-Petersson metric
whose K\"ahler potential $K$ is determined from
\begin{equation}\label{Kahlerpotential}
e^{-K(z,\bar{z})}=i\int_{\mathcal{X}_{z}} \Omega_{z}\wedge
\overline{\Omega}_{z}\,,
\end{equation}
where $\Omega=\{\Omega_{z}\}$ is a section of the Hodge line bundle
$\mathcal{L}$.
 The metric
$G_{i\bar{\jmath}}=\bar{\partial}_{\bar{\jmath}}\partial_{i}K$ is
the Hodge metric induced from the Hermitian metric
$h(\Omega,\Omega)=i^{3}\int \Omega\wedge \overline{\Omega}$ on the
Hodge line bundle $\mathcal{L}$. This metric is called special
K\"ahler metric \cite{Strominger:1990pd, Freed:1999sm}. Among its
other properties, it satisfies the following ``special geometry
relation''
\begin{equation}\label{specialgeometryrelation}
-R_{i\bar{j}~l}^{~~k}=\partial_{\bar{\jmath}}\Gamma_{il}^{k}=\delta^{k}_{l}G_{i\bar{\jmath}}+\delta^{k}_{i}G_{l\bar{\jmath}}
-C_{ilm}\bar{C}_{\bar{\jmath}}^{mk},\quad
i,\bar{\jmath},k,l=1,2,\cdots \dim \mathcal{M}\,,
\end{equation}
where
\begin{equation}\label{yukawa}
C_{ijk}(z)=-\int_{\mathcal{X}_{z}}\Omega_{z}\wedge
\partial_{i}\partial_{j}\partial_{k}\Omega_{z}
\end{equation}
is the so-called \textbf{Yukawa coupling} and
$\bar{C}_{\bar{\jmath}}^{mk}=e^{2K}G^{k\bar{k}}G^{m\bar{m}}\overline{C}_{\bar{j}\bar{k}\bar{m}}$.
There is a natural covariant derivative $D_{i}$ acting on sections
of the Hodge bundle
$\mathcal{H}=\mathcal{R}^{3}\pi_{*}\underline{\mathbb{C}}\otimes
\mathcal{O}_{\mathcal{M}}$ which is the sum of the Chern connection
associated to the Weil-Petersson metric and the connection on
$\mathcal{L}$ induced by the Hermitian metric $h$. See
\cite{Strominger:1990pd, Freed:1999sm, Hosono:2008ve} for details on
this.

The Yukawa couplings will play a key role in the construction of the
rings $(\mathcal{R},\widetilde{\mathcal{R}},\widehat{\mathcal{R}})$
in this paper.

\subsection{Holomorphic anomaly equations}
The genus $g$ topological string partition function
$\mathcal{F}^{(g)}$ as defined in ~\cite{Bershadsky:1993ta,
Bershadsky:1993cx} is a section of the line bundle
$\mathcal{L}^{2-2g}$ over $\mathcal M$, it is shown to satisfy the
following holomorphic anomaly equation:
\begin{eqnarray}
\bar{\partial}_{\bar{\imath}} \partial_{j}\mathcal{F}^{(1)} &=&
\frac{1}{2} C_{jkl} \overline{C}^{kl}_{\bar{\imath}}+
(1-\frac{\chi}{24}) G_{j
\bar{\imath}}\,,\label{haeforgenusone}\\
\bar{\partial}_{\bar{\imath}} \mathcal{F}^{(g)} &=&{1\over 2}
\overline{C}_{\bar{\imath}}^{jk} \left( \sum_{r=1}^{g-1}
D_j\mathcal{F}^{(r)} D_k\mathcal{F}^{(g-r)} +
D_jD_k\mathcal{F}^{(g-1)} \right), \quad g\geq 2\,,
\end{eqnarray}
where $\chi$ is the Euler characteristic of the mirror manifold of
the Calabi-Yau threefold $X$. Any higher genus $\mathcal{F}^{(g)}$
can be determined recursively from these up to addition by a
holomorphic function.

A solution of these recursion equations is given in terms of Feynman
rules \cite{Bershadsky:1993cx}. The propagators $S^{ij}$, $S^i$, $S$
for these Feynman rules are defined by the following:
\begin{equation}
\partial_{\bar{\imath}} S^{ij}= \overline{C}_{\bar{\imath}}^{ij}, \qquad
\partial_{\bar{\imath}} S^j = G_{i\bar{\imath}} S^{ij}, \qquad
\partial_{\bar{\imath}} S = G_{i \bar{\imath}} S^i.
\label{prop}
\end{equation}
By definition, the propagators $S$, $S^i$ and $S^{ij}$ are sections
of the bundles $\mathcal{L}^{-2}\otimes \text{Sym}^m T\mathcal{M}$
with $m=0,1,2$, respectively. The vertices of the Feynman rules are
given by the functions $\mathcal{F}^{(g)}_{i_1\cdots
i_n}=D_{i_{1}}\cdots D_{i_{n}}\mathcal{F}^{(g)}$.

In~\cite{Yamaguchi:2004bt, Alim:2007qj} it was proved, using the
special geometry relation (\ref{specialgeometryrelation}) and the
holomorphic anomaly equation for genus one (\ref{haeforgenusone}),
that the holomorphic anomaly equations for $g\geq 2$ can be put into
the following form
\begin{equation}\label{polynomialsol}
\bar{\partial}_{\bar{i}}\mathcal{F}^{(g)}=\bar{\partial}_{\bar{i}}\mathcal{P}^{(g)}\,,
\end{equation}
where $\mathcal{P}^{(g)}$ is a polynomial of the generators
$S^{ij},S^{i},S,K_{i}$ with the coefficients being holomorphic
quantities. The non-holomorphicity of the topological string
partition functions only comes from these generators and the
anti-holomorphic derivative on the left-hand side of the holomorphic
anomaly equations can be replaced by derivatives with respect to
these generators. Moreover, the derivatives of the generators are
given by \cite{Alim:2007qj}
\begin{eqnarray}\label{polyring}
D_i S^{jk} &=&\delta_i^j S^k+\delta_{i}^{k}S^{j}
-C_{imn} S^{mj} S^{nk}  + h_i^{jk} \, , \nonumber\\
D_i S^j &=& -C_{imn}S^{m}S^{jn}+2\delta^j_i S +h_{i}^{jk}K_{k}+h_i^j\,,\nonumber\\
D_i S&=& -\frac{1}{2} C_{imn} S^m S^n+{1\over 2}h_{i}^{mn}K_{m}K_{n}+h_i^{m}K_{m}+h_{i} \, ,\nonumber\\
D_i K_j &=& -K_i K_j +C_{ijm} S^{mn}K_{n} -C_{ijm}S^{m}+ h_{ij} \, ,
\end{eqnarray}
where $h_i^{jk},h_i^j,h_{i},h_{ij}$ are holomorphic functions. These
holomorphic functions can not be uniquely determined since the above
equations are derived by integrating some equations
\cite{Alim:2007qj, Alim:2008kp}, hence (\ref{polyring}) does not
actually give a ring due to the existence of these holomorphic
functions. To make it a real ring, one needs to include of all of
the derivatives of these holomorphic ambiguities
\cite{Hosono:2008ve}. In \cite{Yamaguchi:2004bt, Alim:2013}, the
holomorphic functions are packaged together by making use of the
special K\"ahler geometry and are essentially parametrized by the
Yukawa couplings. Then one gets a differential ring
$\widehat{\mathcal{R}}$ with finite generators. It also has a nice
grading serving as the ``modular weights".

Note that (\ref{polynomialsol}) only determines the quantity
$\mathcal{F}^{(g)}$ up to addition by an holomorphic quantity $f^{(g)}$ called the genus $g$ holomorphic ambiguity
\begin{equation}
\mathcal{F}^{(g)}=\mathcal{P}^{(g)}(S^{ij},S^{i},S,K_{i})+f^{(g)}\,.
\end{equation}
Boundary conditions are needed to fix the holomorphic ambiguity
$f^{(g)}$, these includes the boundary condition at the large complex
structure limit and the gap condition at the conifold loci on
$\mathcal{M}$ \cite{Bershadsky:1993ta,
Bershadsky:1993cx,Ghoshal:1995wm,Huang:2006si,Huang:2006hq}.
Recently, in \cite{Alim:2013} it was observed that for the
one-parameter non-compact Calabi-Yau threefold family given in
(\ref{PFfornoncompactCY}), the deformation space $\mathcal{M}$ could
be identified with certain modular curve $X_{0}(N)$, the large
complex structure limit and conifold point are the two cusp classes
$[i\infty]$ and $[0]$. Moreover, the ring $\widehat{\mathcal{R}}$ of
propagators is essentially equivalent to the ring of
almost-holomorphic modular forms $\widehat{M}_{*}(\Gamma_{0}(N))$.
It follows by polynomial recursion \cite{Bershadsky:1993cx,
Yamaguchi:2004bt, Alim:2007qj} that the partition functions
$\mathcal{F}^{(g)}$ are almost-holomorphic modular functions (i.e.,
modular weights are zero). The boundary conditions then become
regularity conditions of these modular objects at the cusps. This
then allows one to get the topological string partition functions as
almost-holomorphic modular functions by recursively solving the
holomorphic anomaly equations with the boundary conditions genus by
genus.


\section{\textup{Differential rings from Picard-Fuchs
equations}}\label{ringsfromPF}

In this section, first we shall construct
$M_{*}(\Gamma_{0}(N)),\widetilde{M}_{*}(\Gamma_{0}(N))$ for the
family $\pi_{\Gamma_{0}(N)}:\mathcal{E}_{\Gamma_{0}(N)}\rightarrow
X_{0}(N)=\Gamma_{0}(N)\backslash \mathcal{H}^{*}$ using its
Picard-Fuchs equation. The procedure then motivates us to construct
similar rings $\mathcal{R},\widetilde{\mathcal{R}}$ for 
one-parameter Calabi-Yau threefold families $\pi:
\mathcal{X}\rightarrow \mathcal{M}$.

\subsection{Picard-Fuchs equations for elliptic curve families}

For the elliptic curve families
$\pi_{\Gamma_{0}(N)}:\mathcal{E}_{\Gamma_{0}(N)}\rightarrow
X_{0}(N), N=1^{*},2,3,4$ whose Picard-Fuchs equations are given by
(\ref{PFforonefold}), the rings of quasi-modular forms can be
constructed explicitly using periods of these elliptic curve
families, see e.g., \cite{Zagier:2008, Maier:2009, Alim:2013} and
references therein. More precisely, for the Picard-Fuchs equation
(\ref{PFforonefold}), we define
\begin{equation}
\beta:=1-\alpha\,,
\end{equation}
and choose a basis of solutions to be
\begin{equation}\label{periodsofcurve}
\omega_{0}=~_{2}F_{1}(1/r,1-1/r,1;\alpha),\quad \omega_{1}={i\over
\sqrt{N}}~_{2}F_{1}(1/r,1-1/r,1;\beta)\,,
\end{equation}
then one has
\begin{equation}
\tau={\omega_{1}\over
\omega_{0}}={i\over\sqrt{N}}{~_{2}F_{1}(1/r,1+1-r,1;\beta)\over
~_{2}F_{1}(1/r,1+1-r,1;\alpha)}\,,
\end{equation}
where $r=6,4,3,2$ for $N=1^{*},2,3,4$ respectively. One then defines
a triple following \cite{Borwein:1991, Berndt:1995, Maier:2009},
\begin{equation}
A=\omega_{0},\quad B=(1-\alpha)^{1\over r}A,\quad C=\alpha^{1\over
r}A\,.
\end{equation}
\begin{rem}
These functions $A,B,C$ are possibly multi-valued on the modular
curve $X_{0}(N)$ and have divisors being ${1\over r}(\alpha=\infty),
{1\over r}(\alpha=1), {1\over r}(\alpha=0)$, respectively. This fact
is not used in this paper, but is useful \cite{Alim:2013} in
analyzing the singularities of the topological string partition
functions as solutions to the holomorphic anomaly equations.
\end{rem}
Note that the quantities $A,B,C$ satisfy the equation
$A^{r}=B^{r}+C^{r}$. We now define further the quantity
\begin{equation}
E=\partial_{\tau}\log {C^{r}B^{r}},\quad \partial_{\tau}:={1\over
2\pi i}{\partial\over
\partial \tau}\,.
\end{equation}
It turns out that the ring generated by $A,B,C,E$ is closed under
the derivative $\partial_{\tau}$.
\begin{thm}\label{quasiringforcurvefamily}
For each of the elliptic curve families
$\pi_{\Gamma_{0}(N)}:\mathcal{E}_{\Gamma_{0}(N)}\rightarrow
X_{0}(N), N=1^{*},2,3,4$ with $r=6,4,3,2$ respectively, the
following identities hold:
\begin{eqnarray}
\partial_{\tau}A&=&{1\over 2r}A(E+{C^{r}-B^{r}\over A^{r}}A^{2})\,,\nonumber\\
\partial_{\tau}B&=&{1\over 2r}B(E-A^{2})\,,\nonumber\\
\partial_{\tau}C&=&{1\over 2r}C(E+A^{2})\,,\nonumber\\
\partial_{\tau}E&=&{1\over 2r}(E^{2}-A^{4})\,.
\end{eqnarray}
\end{thm}
\begin{proof}
 It follows from (\ref{DBC}), (\ref{DA}),
 (\ref{DE}) below.
\end{proof}

The ring generated by $A,B,C,E$ has a obvious grading denoted by $k$
below: the gradings assigned to $A,B,C,E$ are $1,1,1,2$, taking the
derivative $\partial_{\tau}$ will increase the grading by $2$.
Similar to the full modular group case, one gets the following

\begin{thm}\label{almostholringforcurvefamily}
For each of the elliptic curve families
$\pi_{\Gamma_{0}(N)}:\mathcal{E}_{\Gamma_{0}(N)}\rightarrow
X_{0}(N), N=1^{*},2,3,4$ with $r=6,4,3,2$ respectively, define
$\hat{E}=E+{r\over 6}{-3\over \mathrm{Im}\tau}$ and
$\hat{\partial_{\tau}}=\partial_{\tau}+{k\over 12}{-3\over
\mathrm{Im}\tau}$, then the following identities hold:
\begin{eqnarray}
\hat{\partial_{\tau}}A&=&{1\over 2r}A(\hat{E}+{C^{r}-B^{r}\over A^{r}}A^{2})\,,\nonumber\\
\hat{\partial_{\tau}}B&=&{1\over 2r}B(\hat{E}-A^{2})\,,\nonumber\\
\hat{\partial_{\tau}}C&=&{1\over 2r}C(\hat{E}+A^{2})\,,\nonumber\\
\hat{\partial_{\tau}}\hat{E}&=&{1\over 2r}(\hat{E}^{2}-A^{4})\,.
\end{eqnarray}
\end{thm}
\begin{proof}
Assume that the desired non-holomorphic quantity $\hat{E}$ is given
by $\hat{E}=E+ \Delta E$ with $\Delta E=\lambda{-3\over
\mathrm{Im}\tau}$ for some constant $\lambda$. Then it is easy to
see with $\hat{\partial_{\tau}}=\partial_{\tau}+{k\over 2r}\Delta
E$, the first three identities follow from Theorem
\ref{quasiringforcurvefamily}. Solving the constant $\lambda$ from
the last identity, we then get $\lambda={r\over 6}$. Thus the
conclusion follows.
\end{proof}
\begin{rem}
Theorems \ref{quasiringforcurvefamily} is known in the literature,
see e.g., \cite{Maier:2011} and references therein. In fact, for
each of these cases, one can find the $\theta$ or $\eta$ expressions
of the quantities $A,B,C,E$ and prove the formulas by checking the
$\theta$ or $\eta$ expressions. The relations between these
generators and the Eisenstein series $E_{2},E_{4},E_{6}$ are also
known. One could then, for example, use the Eisenstein series
expression of $E$ to obtain the almost-holomorphic modular form
$\hat{E}=E+\Delta E$, with
\begin{equation}
\Delta E={2\over N+1}{-3\over \mathrm{Im}\tau}, \,N=1^{*},2,3,\quad
\Delta E={1\over 3}{-3\over \mathrm{Im}\tau},\, N=4\,.
\end{equation}
These agree with the above choices $\Delta E={r\over 6}{-3\over
\textrm{Im}\tau}$. See e.g., \cite{Alim:2013} for a collection of
these results.

But this method could not be generalized to Calabi-Yau threefold
families.
\end{rem}
\begin{rem}\label{reducedring}
Strictly speaking, the ring generated by $A,B,C,E$ above do not form
a differential ring due to the negative powers in the equations.
However, it is easy to see by choosing suitable powers of these
generators, one can indeed get a ring. For example, in the
$r=6,4,3,2$ cases, one can choose $A^{4},B^{6}-C^{6},E;
A^{2},B^{4},E; A,B^{3}$ and $E;A,B^{2},E$, respectively.

The ring generated by $A,B,C,E$ is not the ring of quasi-modular
forms for $\Gamma_{0}(N)$. For example, in the case $N=3$, the ring
of quasi-modular forms with the nontrivial multiplier system
$\chi_{-3}$ is
$\widetilde{M}_{*}(\Gamma_{0}(3),\chi_{-3})=\mathbb{C}[A,B^{3},E]\cong
\mathbb{C}[A,F=C^{3}-B^{3},E] $ and the differential equations
are\footnote{The author would like to thank Prof. Don Zagier for
discussions on these.}
\begin{eqnarray*}
\partial_{\tau}A&=&{1\over 2}(EA+F)\,,\\
\partial_{\tau}F&=&{1\over 2}(EF+A^{5})\,,\\
\partial_{\tau}E&=&{1\over 6}(E^{2}-A^{4})\,.
\end{eqnarray*}
One also has
$\widetilde{M}_{*}(\Gamma_{0}(2))=\mathbb{C}[A^{2},B^{4},E]$,
$\widetilde{M}_{*}(\Gamma_{0}(4),\chi_{-4})=\mathbb{C}[A,B^{2},E]$,
etc.

However, in the whole discussion of this work we shall not the
transformations in the corresponding modular group $\Gamma_{0}(N)$,
but use only the differential equations they satisfy. Moreover, in
the applications to the study of the holomorphic anomaly equations
and the topological string partition functions, eventually only
elements in the above rings of quasi-modular forms ($r=6$ case is
exceptional) will be involved \cite{Alim:2013}.

For the reasons mentioned above, by abuse of langauge, we shall call
the rings $\mathbb{C}[A^{\pm 1},B^{\pm 1},C^{\pm
1}],\mathbb{C}[A^{\pm 1},B^{\pm 1},C^{\pm 1},E]$ and
$\mathbb{C}[A^{\pm 1},B^{\pm 1},C^{\pm 1},\hat{E}]$ the ring of
modular forms, quasi-modular forms, almost-holomorphic modular forms
for $\Gamma_{0}(N)$, and denote them by $M_{*}(\Gamma_{0}(N))$,
$\widetilde{M}_{*}(\Gamma_{0}(N))$, $\widehat{M}_{*}(\Gamma_{0}(N))$
respectively. We shall also call the gradings ``modular weights"
which could be negative.
\end{rem}

Since later we shall need to generalize the construction to some
Calabi-Yau threefold families using the Picard-Fuchs equations, we
shall reproduce below the details in constructing the graded
differential rings
$(\widetilde{M}_{*}(\Gamma_{0}(N)),\partial_{\tau})$ using
properties of the equations in (\ref{PFforonefold}).

We start from the following observation.
\begin{prop}\label{DalphabyPF}
For each of the elliptic curve families
$\pi_{\Gamma_{0}(N)}:\mathcal{E}_{\Gamma_{0}(N)}\rightarrow
X_{0}(N), N=1^{*},2,3,4$, one has $\partial_{\tau}\alpha=\alpha\beta
A^{2}$, where as before $\beta:1-\alpha, \partial_{\tau}={1\over
2\pi i}{\partial\over \partial \tau}$.
\end{prop}
\begin{proof}
For ease of notation, first we write the Picard-Fuchs operator in
(\ref{PFforonefold}) as
\begin{equation}\label{Pfforonefoldgeneral}
\mathcal{L}=\theta^{2}-\alpha (\theta+c_{1})(\theta
+c_{2}),~~\textrm{with}~~c_{1}={1\over r},\,c_{2}=1-{1\over r}\,,
\end{equation}
and define
\begin{equation*}
\tilde{\mathcal{L}}:=\left((\theta+\theta \log\omega_{0})^{2}-\alpha
(\theta+\theta \log \omega_{0}+c_{1})(\theta +\theta \log
\omega_{0}+c_{2})\right),
\end{equation*}
Then we have
\begin{equation*}
\omega_{0}~\tilde{\mathcal{L}}\,{\Pi\over
\omega_{0}}=\mathcal{L}\,\Pi=0~\, \textrm{for a period}~\Pi\,.
\end{equation*}
In particular, we have $\tilde{\mathcal{L}}\, 1=0,\,
\tilde{\mathcal{L}}\, \tau=0\,.$ Subtracting $\tilde{\mathcal{L}}\,
1$ from $\tilde{\mathcal{L}}\, \tau$, one then gets
\begin{equation*}\label{theta2tau}
\beta\theta^{2}\tau+\left(2\beta\theta \log
\omega_{0}-\alpha(c_{1}+c_{2})\right)\theta\tau=0\,.
\end{equation*}
That is,
\begin{equation*}
\theta\log (\omega_{0}^{2}\,\theta \tau )-(c_{1}+c_{2}){\alpha\over
\beta}= \theta\log (\omega_{0}^{2}\,\theta \tau
)+(c_{1}+c_{2})\theta \log  \beta =0\,.
\end{equation*}
Solving this first order equation for $\omega_{0}^{2}\,\theta \tau$,
we then get
\begin{equation*}\label{thetatau}
\theta \tau = {c\over \beta^{c_{1}+c_{2}}\omega_{0}^{2}}= {c\over
\beta\omega_{0}^{2}}\,.
\end{equation*}
for some constant $c$. By looking at the leading terms in $\alpha$
of both sides as $\alpha\rightarrow 0$, we can then find that
$c={1\over 2\pi i}$. Hence $\partial_{\tau}\alpha =\alpha\beta
A^{2}$ as claimed.
\end{proof}

In what follows, we shall call the modular function $\alpha$ the
algebraic modulus for the modular curve, while $\tau$ the
transcendental modulus for the modular curve. The above formula then
gives a differential equation relating the algebraic and
transcendental moduli.

Recall the definitions of $B,C$ which implies that
${\alpha\over\beta}={C^{r}\over B^{r}}$, we then get
\begin{cor}
For each of the elliptic curve families
$\pi_{\Gamma_{0}(N)}:\mathcal{E}_{\Gamma_{0}(N)}\rightarrow
X_{0}(N), N=1^{*},2,3,4$, the following is true:
\begin{equation}
 A^{2}={\partial_{\tau}\alpha\over
\alpha\beta}=\partial_{\tau}\log {\alpha\over
\beta}=\partial_{\tau}\log {C^{r}\over B^{r}}\,.
\end{equation}
\end{cor}
Using the definition $E=D\log {C^{r} B^{r}}$, we have
\begin{equation}\label{DBC}
\partial_{\tau}B={1\over2r}B (E-A^{2}),\quad \partial_{\tau}C={1\over 2r}C(E+A^{2})\,.
\end{equation}
From $A^{r}=B^{r}+C^{r}$, we can easily get
\begin{equation}\label{DA}
\partial_{\tau}A={1\over 2r}A(E+{C^{r}-B^{r}\over A^{r}}A^{2})\,.
\end{equation}
Using the Picard-Fuchs equation satisfied by $A$:
\begin{equation*}
\beta (\theta A)^{2}-(c_{1}+c_{2})\alpha\theta A-c_{1}c_{2}\alpha
A=0\,,
\end{equation*}
we obtain
\begin{equation*}
\partial_{\tau}^{2}\log A=(\partial_{\tau}\log A)^{2}+(c_{1}+c_{2}-1)\alpha\beta A^{2}\partial_{\tau}\log
A+c_{1}c_{2}\alpha\beta A^{4}\,,
\end{equation*}
This second order differential equation of $A$ will become a first
order differential equation of $E$ since (\ref{DA}) says that
$E=2r\partial_{\tau}\log A-(\alpha-\beta)A^{2}$, one then gets
\begin{equation}\label{DE}
\partial_{\tau}E={1\over 2r}E^{2}+ (2c_{1}c_{2}r\alpha\beta -{1\over
2r}(\alpha-\beta)^{2}-2\alpha\beta)A^{4} = {1\over 2r}E^{2}-{1\over
2r}(\alpha+\beta)^{2}A^{4}={1\over 2r}(E^{2}-A^{4})\,.
\end{equation}

\subsubsection*{``Special coordinate'' on the deformation space of
elliptic curves}
In the previous discussions, we obtained the graded differential
rings started from Prop. \ref{DalphabyPF} which could be thought
of as a differential equation satisfied by the normalized period of
the elliptic curves. As we shall see, for a one-parameter Calabi-Yau
threefold family $\pi:\mathcal{X}\rightarrow\mathcal{M}$, the
``normalized period'', called special coordinate $t$ below, does not
satisfy an analogous relation. However, there is a parameter $\tau$
on $\mathcal{M}$ satisfying a similar identity (see Prop.
\ref{analogueofDalpha} below). To get the identity, one computes the
Yukawa couplings in different coordinates.

In the following, we shall explain how to derive Prop.
\ref{DalphabyPF} by computing the Yukawa couplings for the purpose
of later generalization.
\begin{prop}\label{Yukawaintau}
The Yukawa coupling in the transcendental modulus $\tau$, defined by
$C_{\tau}=\int \Omega\wedge
\partial_{\tau}\Omega$, satisfies $C_{\tau}=1$.
\end{prop}
That is, there is no ``quantum correction'' \cite{Lian:1994zv,
Lian:1995, Lian:1996, Hosono:2008ve} added to the classical
intersection number $\kappa=1$.
\begin{proof}
Take the local parameter of the (punctured) deformation space near
the point $\alpha=0$ to be $\tau$. At the base point $\tau_{*}$, the
fiber of the family $\pi_{\Gamma_{0}(N)}:
\mathcal{E}_{\Gamma_{0}(N)}\rightarrow X_{0}(N)$ is the torus
$T_{\tau_{*}}=\mathbb{C}/\Lambda_{\tau_{*}}$, where
$\Lambda_{\tau_{*}}=\mathbb{Z}\oplus \mathbb{Z}\tau_{*}$. We take
the holomorphic top form on $T_{\tau_{*}}$ to be $dz_{\tau_{*}}$
descended from $\mathbb{C}$. For any $\tau$ near the base point
$\tau_{*}$, the diffeomorphism sending $T_{\tau_{*}}$ to $T_{\tau}$
is given by $z_{\tau}={\tau-\bar{\tau}_{*}\over
\tau_{*}-\bar{\tau}_{*}}z_{\tau_{*}}+{\tau_{*}-\tau\over
\tau_{*}-\bar{\tau}_{*}}\bar{z}_{\tau_{*}}\,.$ From this, one can
then see that $\partial_{\tau}\Omega_{\tau}|_{\tau_{*}}={1\over
\tau_{*}-\bar{\tau}_{*}}(dz_{\tau_{*}}-d\bar{z}_{\tau_{*}})\,.$ It
follows then that $C_{\tau_{*}}=\int_{T_{\tau_{*}}}{i\over
2\mathrm{Im}\tau_{*}}dz_{\tau_{*}}\wedge d\bar{z}_{\tau_{*}}=1\,.$
\end{proof}
One can also compute the Yukawa coupling in the algebraic modulus
$\alpha$ as follows.
\begin{prop}\label{Yukawainalpha}
The Yukawa coupling in the algebraic modulus $\alpha$, defined by
$C_{\alpha}=\int \Omega\wedge \partial_{\alpha}\Omega\,,$ satisfies
$C_{\alpha}={1\over \alpha\beta}\,.$
\end{prop}
\begin{proof}
Recall that the Picard-Fuchs equation (\ref{Pfforonefoldgeneral})
tells that when integrated over cycles, one has
$\theta^{2}\Omega=(c_{1}+c_{2}){\alpha\over \beta}\theta
\Omega+c_{1}c_{2}{\alpha\over \beta}\Omega$. Now we have
\begin{eqnarray*}
\theta(\alpha C_{\alpha}) &=&\theta\int \Omega\wedge \theta\Omega
=\int \theta\Omega\wedge \theta\Omega+\int \Omega\wedge \theta^{2}\Omega\\
&=&0+\int \Omega\wedge \left((c_{1}+c_{2}){\alpha\over \beta}\theta \Omega+c_{1}c_{2}{\alpha\over \beta}\Omega\right)\\
&=&0+ (c_{1}+c_{2}){\alpha\over \beta}\int \Omega\wedge\theta \Omega+0\\
&=&(c_{1}+c_{2}){\alpha\over \beta}(\alpha C_{\alpha})={\alpha\over \beta}(\alpha C_{\alpha})\\
\end{eqnarray*}
Solving $\alpha C_{\alpha}$ from this equation, we get $\alpha
C_{\alpha}={c\over \beta}$ from some constant $c$. We then fix this
$c$ by looking at the behavior of both sides near $\alpha=0$. This
gives $c=1$. Hence the conclusion follows.
\end{proof}

By computing the Yukawa coupling in two different coordinates $\tau$
and $\alpha$ in Prop. \ref{Yukawaintau} and Prop.
\ref{Yukawainalpha}, we can then derive the equation
$\partial_{\tau}\alpha=(C_{\alpha})^{-1}\omega_{0}^{2}$ given in
Prop. \ref{DalphabyPF} between the transcendental modulus $\tau$ and
the algebraic modulus $\alpha$, from the following relation
\begin{equation}
C_{\tau}={1\over 2\pi i}{1\over \omega_{0}^{2}}{\partial \alpha\over
\partial \tau}C_{\alpha}\,.
\end{equation}

\subsection{Picard-Fuchs equations for Calabi-Yau three--fold families}

In this section, motivated by the discussions in elliptic curve
families, we shall work out similar rings
$\mathcal{R},\widetilde{\mathcal{R}}$ living on the deformation
spaces of Calabi-Yau three--folds. As before we limit ourselves to
the case $\dim \mathcal{M}=1$. We shall start by computing the
Yukawa couplings in different coordinates, then we derive an
equation analogous to Prop. \ref{DalphabyPF} between the complex
coordinate $\alpha$ and a suitably chosen coordinate $\tau$ on the
deformation space $\mathcal{M}$. After that we construct a ring out
of special K\"ahler geometry quantities (connections, Yukawa
couplings, etc.).

In the following we shall first consider slightly more general
Picard-Fuchs equations before we specialize to the Picard-Fuchs
(\ref{PFfornoncompactCY}) (\ref{PFforcompactCY}) mentioned above.

Suppose the Picard-Fuchs operator for the family $\pi:
\mathcal{X}\rightarrow \mathcal{M}$ of Calabi-Yau threefolds $X$ is of the form
\begin{equation}\label{PFforcompactCYgeneral}
\mathcal{L}=\theta^{4}-\alpha
\prod_{i=1}^{4}(\theta+c_{i})=(1-\alpha)\theta^{4}-\alpha(\sigma_{1}\theta^{3}+\sigma_{2}\theta^{2}+\sigma_{3}\theta+\sigma_{4})\,,
\end{equation}
where $\theta=\alpha{\partial\over \partial\alpha}$ and
$\sigma_{i}$s are the symmetric polynomials of the constants
$c_{1},c_{2},c_{3},c_{4}$. For the quintic mirror family case, one
has $(c_{1},c_{2},c_{3},c_{4})=(1/5,2/5,3/5,4/5)$ and thus
$\sigma_{1}=2$. As before, we shall denote $\beta:=1-\alpha$.

The large complex structure limit is given by $\alpha=0$. Near this
point, the solutions to the Picard-Fuchs equation $\mathcal{L}\,
\Pi=0$ could be obtained by the Frobenius method and have the
following form
\begin{equation}\label{soltoPFforcompactCY}
(X^{0},X^{1},P_{1},P_{0})=X^{0}(1,t,F_{t},2F-tF_{t})\,
\end{equation}
where $X^{0}(\alpha)\sim 1+\mathcal{O}(\alpha),\, t\sim \log
\alpha+\cdots$ near $\alpha=0$, and $F_{t}=\partial_{t}F(t)$. The
function $F(t)$, called prepotential, has the form from mirror
symmetry
\begin{equation}\label{prepotential}
F(t)={\kappa\over
3!}t^{3}+at^{2}+bt+c+\sum_{d=1}^{\infty}N_{0,d}^{\textrm{GW}}e^{dt}\,,
\end{equation}
where $\kappa$ is the classical triple intersection number of the
mirror manifold $\check{X}$ of $X$, $a,b,c$ are some numbers which
are not important in our discussions, and $N^{\textrm{GW}}_{0,d}$ is
the genus zero degree $d$ Gromov-Witten invariant of $\check{X}$.
This particular structure among the periods comes from the special
K\"ahler geometry on $\mathcal{M}$ and mirror symmetry.

The Yukawa coupling in the $t$ coordinate is then given by
$C_{ttt}=F_{ttt}=\kappa+\mathcal{O}(q^{1})$ with $q=e^{t}$. In the
complex coordinate $\alpha$, we have
\begin{prop}
The Yukawa coupling, defined by $C_{\alpha\alpha\alpha}=-\int
\Omega\wedge \partial_{\alpha}^{3}\Omega$, is given by
$C_{\alpha\alpha\alpha}={\kappa\over \alpha^{3}\beta}$.
\end{prop}
\begin{proof}
First due to Griffiths transversality, we have
$\alpha^{3}C_{\alpha\alpha\alpha}=-\int \Omega\wedge
\theta^{3}\Omega.$ By integration by parts and Griffiths
transversality, it follows that
\begin{eqnarray*}
\theta(\alpha^{3}C_{\alpha\alpha\alpha})&=&-\int \theta \Omega\wedge \theta^{3}\Omega-\int \Omega\wedge \theta^{4}\Omega\\
&=&-\left(\theta\int \theta \Omega\wedge \theta^{2}\Omega-\int \theta^{2} \Omega\wedge \theta^{2}\Omega\right)-\int \Omega\wedge (\sigma_{1}{\alpha\over \beta}\theta^{3}\Omega+\cdots)\\
&=&-\theta\left(\theta \int \Omega\wedge \theta^{3}\Omega- \int \Omega\wedge \theta^{3}\Omega\right)-0+\sigma_{1}{\alpha\over\beta}(\alpha^{3}C_{\alpha\alpha\alpha})\\
&=&-\theta
(\alpha^{3}C_{\alpha\alpha\alpha})+\sigma_{1}{\alpha\over\beta}(\alpha^{3}C_{\alpha\alpha\alpha})\,.
\end{eqnarray*}
Solving $\alpha^{3}C_{\alpha\alpha\alpha}$ from this equation, we
then get
$C_{\alpha\alpha\alpha}={c\over \alpha^{3}\beta}$.
Using the fact that
\begin{equation}\label{Cttt}
{1\over (X^{0})^{2}}C_{\alpha\alpha\alpha}({\partial\alpha\over
\partial t})^{3}=C_{ttt}=\kappa+\mathcal{O}(q)\,,
\end{equation} we know $c=\kappa$.
Hence the assertion follows.
\end{proof}

Now that we have computed the Yukawa coupling in the special
coordinate $t$ and complex coordinate $\alpha$, we shall find the
analogue of Prop. \ref{DalphabyPF} by defining the following
coordinate
\begin{equation}\label{defoftau}
\tau={1\over 2\pi i}\kappa^{-1}F_{tt}\,.
\end{equation}
This definition was motivated \cite{Aganagic:2006wq} to establish
the modularity for non-compact Calabi-Yau threefolds, as we shall
explain later. According to this definition of $\tau$, near
$\alpha=0$ we have $\tau(\alpha)\sim{1\over 2\pi i} \log
\alpha+\mathcal{O}(\alpha^{0})$. Therefore, from
\begin{equation}
{1\over (X^{0})^{2}}C_{\alpha\alpha\alpha} ({\partial \alpha\over
\partial t})^{3}=C_{ttt}=2\pi i \kappa {\partial \tau\over
\partial t}=2\pi i \kappa {\partial \tau\over \partial \alpha}{\partial
\alpha\over \partial t}\,.
\end{equation}
we obtain the following assertion.
\begin{prop}\label{analogueofDalpha}
\begin{equation}
D\alpha=\alpha \cdot
\kappa(\alpha^{3}C_{\alpha\alpha\alpha})^{-1}\cdot  (X^{0}\theta
t)^{2}=\alpha\beta   (X^{0}\theta t)^{2},\quad D:={1\over 2\pi
i}{\partial\over \partial \tau}\,.
\end{equation}
\end{prop}
Note that the only places in which we have used the special K\"ahler
geometry are in the definition (\ref{defoftau}) of $\tau$ in terms
of $F_{tt}$ and the limit of (\ref{Cttt}) as $\alpha$ goes to $0$.
But we could have defined $\tau$ as the quantity satisfying the
equation ${1\over (X^{0})^{2}}C_{\alpha\alpha\alpha} ({\partial
\alpha\over
\partial t})^{3}=2\pi i  \kappa{\partial \tau\over
\partial t}$ and the condition $\lim_{\alpha\rightarrow 0}2\pi i  {\partial \tau\over
\partial t}=1$
 without referring to the prepotential $F(t)$ and Yukawa
coupling $C_{ttt}$, thus only the periods and no special K\"ahler geometry are used.\\

We shall now take Prop. \ref{analogueofDalpha} as the starting point
to construct the analogue of the ring of quasi-modular forms.
Motivated by the discussions of elliptic curve families, we define
the following triple
\begin{equation}
A=X^{0}\theta t, \quad B=(1-\alpha)^{1\over r}A,\quad
C=\alpha^{1\over r}A\,,
\end{equation}
where $r$ is some undetermined constant and does not show up in the
final form of the ring $\widehat{\mathcal{R}}$ we shall consider
later.
 Similarly we define
\begin{equation}\label{dfofE}
E=D\log{C^{r}B^{r}}=D\log \alpha\beta A^{2r}=(\alpha-\beta)A^{2}+D
\log A^{2r}\,.
\end{equation}
Now thanks to Prop. \ref{analogueofDalpha}, we get
\begin{equation}\label{A2}
A^{2}={D\alpha\over \alpha\beta}=D\log {\alpha\over \beta}=D\log
{C^{r}\over B^{r}}\,.
\end{equation}
We also have the following relations among these generators
following from the definitions of $A,B,C,E$ and (\ref{A2}),
\begin{eqnarray}
DB={1\over 2r}B(E-A^{2}),\quad DC={1\over 2r}C(E+A^{2})\,.
\end{eqnarray}

To get a closed ring, we need to prove $A$ satisfies a differential
equation with coefficients being holomorphic functions of
$\alpha,\beta$. Define
\begin{equation}
A'=X^{0},\quad A''=\theta t\,,
\end{equation}
 it turns out after adding $D^{i}A', D^{i}A'', i=1,2,3$, the ring
will close under the derivative $D$. Note that the generator $E$ is
already contained according to (\ref{dfofE}).
\begin{prop} \label{algebraofquasimodularforms}
The ring $\widetilde{\mathcal{R}}$ generated by $D^{i}A', i=0,1,2,3;
D^{j}A'', j=0,1,2$ and $B,C,B^{-1},C^{-1}$, is closed under the
derivative $D$.
\end{prop}
\begin{proof}
The Picard-Fuchs equation tells that if one defines
\begin{equation*}
\tilde{\mathcal{L}}=(\theta+\theta \log X^{0})^{4}-\alpha
\prod_{i=1}^{4}(\theta+\theta \log X^{0}+c_{i})\,,
\end{equation*}
then $X^{0}\tilde{\mathcal{L}}\, {\Pi\over
X^{0}}=\mathcal{L}\,\Pi=0\,~\textrm{for a period}~\Pi\,.$ In
particular, one has $\mathcal{L}\, {X^{0}}=0$ and
$\tilde{\mathcal{L}}\, {X^{1}\over X^{0}}=\tilde{\mathcal{L}}\,t=0$.
The first equation $\mathcal{L}\, {X^{0}}=0$ tells that
$\theta^{4}X^{0}$ could be expressed as a polynomial of
$\theta^{i}X^{0},i=0,1,2,3$ with coefficients being rational
functions of $\alpha,\beta$. Using the relation $\theta=\beta^{-1}
(X^{0}\theta t)^{-2}D$ following from Prop.
\ref{analogueofDalpha}, we know that $D^{4}X^{0}$ is a polynomial
in $D^{i}X^{0}, 0, i=1,2,3; \,D^{j}\theta t,j=0,1,2,3$ and
$B,C,B^{-1},C^{-1}$. Similarly, by considering the second equation
$\tilde{\mathcal{L}}\,t=0$, one sees that $\theta^{3}\theta t$ and
thus $D^{3}\theta t$ is also contained in the ring as claimed.
\end{proof}
\begin{rem}
Note that when taking the derivative $D$, negative powers of
generators will appear. But as mentioned in Remark
\ref{reducedring}, to avoid them one only needs to choose a suitable
set of generators carefully. In fact, in the final form of the
graded ring $\widehat{\mathcal{R}}$ we shall consider below, we are
going to make a specific choice of generators so that no negative
powers will appear in the derivatives of the generators.
\end{rem}
From Prop. \ref{analogueofDalpha} one can easily see that in fact
the subring generated by $D^{i}A', i=0,1,2,3; D^{j}A'',
j=0,1,2;\alpha^{\pm}, \beta^{\pm}$ is also closed under $D$. We
shall denote this differential subring by
$(\widetilde{\mathcal{R}}^{sub},D)$ in which the constant $r$ does
not show up.

\subsubsection*{Picard-Fuchs equations for non-compact Calabi-Yau
threefold families}\label{quasiringfornoncompactCY}

Now we consider the non-compact Calabi-Yau three--fold families
(\ref{PFfornoncompactCY}) whose Picard-Fuchs equations reduce to
some third order differential equations of the form
$\mathcal{L}_{\textrm{elliptic}}\circ \theta $. For each of these
families, as explained in \cite{Alim:2013}, one can identify the
base $\mathcal{M}$ with a certain modular curve $X_{0}(N)$.

Then one has $X^{0}=1$ and thus $A=\theta t=\omega_{0}$; moreover,
by choosing the normalization for $F_{t}$ suitably, we can make
$\theta F_{t}$ to be $2\pi i\kappa \omega_{1}$, where
$\omega_{0},\omega_{1}$ are the periods of
$\mathcal{L}_{\textrm{elliptic}}$ given in (\ref{periodsofcurve}).
Now the parameter $\tau={1\over 2\pi i}\kappa^{-1}F_{tt}={1\over
2\pi i}\kappa^{-1}{\theta F_{t}\over \theta t}$ is equal to
${\omega_{1}\over \omega_{0}}$, and the parameter $\tau$ is the
transcendental modulus of the elliptic curve sitting inside the
Calabi-Yau threefold and lies in the upper half plane
$\mathcal{H}$.
 Therefore, in
these cases, one has $\mathcal{R}\cong \mathbb{C}[A^{\pm 1},B^{\pm
1},C^{\pm 1}]=M_{*}(\Gamma_{0}(N)),\widetilde{\mathcal{R}}\cong
\mathbb{C}[A^{\pm 1},B^{\pm 1},C^{\pm
1},E]=\widetilde{M}_{*}(\Gamma_{0}(N))$. See \cite{Alim:2013} and
references therein for details.

\subsubsection*{Gradings}There are two natural gradings, denoted by $(k,m)$ on the
ring $\tilde{\mathcal{R}}$. The grading $m$ indicates that the
element is a section of $\mathcal{L}^{m}$ and will be called the
degree. Recall that $X^{0}$ is a period of the form $\int_{C}\Omega$
and $C_{\alpha\alpha\alpha}=-\int_{X}\Omega\wedge
\partial_{\alpha}^{3}\Omega$, where $\Omega$ is a section of the Hodge
line bundle $\mathcal{L}\rightarrow \mathcal{M}$, we can easily
figure out the degree of the generators. The second grading, called
the weight $k$, is motivated by the studies of elliptic curve
families and non-compact Calabi-Yau threefolds discussed above, in
which $\tau$ is really parametrizing the upper half plane
$\mathcal{H}$. We then defines the degrees and weights for the
quantities $X^{0}, \theta t , B,C,\alpha, (\alpha^{3}
C_{\alpha\alpha\alpha})$ to be $(1,0),(0,1),(1,1),(1,1),(0,0),(0,2)$
respectively. Taking the derivative $D$ with respect to $\tau$ will
not change the degree, but raise the weight by $2$. Then we have the
decomposition $\mathcal{R}=\oplus_{(k,m)}\mathcal{R}_{k,m}$.
Similarly, there is a
such decomposition for the graded differential ring $(\widetilde{\mathcal{R}},D)$.\\

The above discussions suggests that the rings
$\mathcal{R}=\mathbb{C}[(X^{0})^{\pm 1},(\theta t)^{\pm 1},B^{\pm
1},C^{\pm 1}]$, $\widetilde{\mathcal{R}}=\mathcal{R}\otimes
\mathbb{C}[D^{i}X^{0},i=1,2,3; D^{j}\theta t,j=1,2]$, defined on the
deformation space $\mathcal{M}$, are the analogues of
$M_{*}(\Gamma),\widetilde{M}_{*}(\Gamma)$ defined on the modular
curve $X_{\Gamma}$, and the weight $k$ plays the role of modular
weight. The generators $D^{i}X^{0},i=1,2,3; D^{j}\theta t,j=1,2$
should be considered as the analogue of quasi-modular forms. We
shall give
more evidences for this later.\\

Similar to what was explained in Remark \ref{reducedring}, one can
get a smaller differential ring $\widetilde{\mathcal{R}}^{sub}$. It
turns out that using special K\"ahler geometry of the deformation
space $\mathcal{M}$, one may further reduce the number of generators
in $\widetilde{\mathcal{R}}-\mathcal{R}$. For example, for the
quintic mirror family case considered, the sequence $D^{i}\theta
t,i=0,1,2$ could be reduced to $D^{i}\theta t,i=0,1$ as discussed in
\cite{Lian:1994zv, Hosono:1996jv, Yamaguchi:2004bt}. This is proved
using the fact that $t$ is the canonical coordinate on the
deformation space $\mathcal{M}$(more than just being the ratio of
two periods), as we shall discuss in the next section.


\section{\textup{Differential rings from special K\"ahler geometry}}\label{ringsfromgeometry}

In this section, we shall use properties of the special K\"ahler
geometry on $\mathcal{M}$ to reduce the number of generators in
$\widetilde{\mathcal{R}}$, and more importantly to define
$\widehat{\mathcal{R}}$ as the ``non-holomorphic completion" of
$\widetilde{\mathcal{R}}$.

We first start by reviewing some basic properties about the
canonical coordinates and holomorphic limits which will be important
later. The discussions on these concepts apply to multi-parameter
Calabi-Yau families.

\subsection{Canonical coordinates and holomorphic limits}

On a Kahler manifold $M$, according to \cite{Bershadsky:1993cx}, the
canonical coordinates $t=\{t^{i}\}_{ i=1,2,\cdots \dim M}$ around
the base point $p$ are defined to be the holomorphic coordinates
such that
\begin{equation}\label{canonicalcoordinates}
\partial_{t^{I}} K_{i}|_{p}=0=\partial_{t^{I}}\Gamma_{ij}^{k}|_{p}\,
\end{equation}
where $I$ is a multi-index and
$\partial_{t^{I}}=\partial_{t^{i_{1}}}\partial_{t^{i_{2}}}\cdots\partial_{t^{i_{m}}}
,\,m=|I|\geq 0.$ Note that the first equation is a condition on the
choice of the K\"ahler potential which transforms under the rule
$K\mapsto K+f+\bar{f}$, where $f$ is purely holomorphic.

These coordinates are studied elsewhere in different contexts, for
example \cite{Kapranov:1999, Higashijima:2001, Higashijima:2002,
Gerasimov:2004yx}. They are the normal coordinates for the K\"ahler
geometry \cite{Higashijima:2001, Higashijima:2002} and can be
constructed using the holomorphic exponential map
\cite{Kapranov:1999}.

\subsubsection*{Exponential map and Gaussian normal coordinates}

Now we shall recall some basic facts from Riemannian geometry. Given
a Riemannian manifold $M$ with the metric $G_{ij}$, the Gaussian
normal coordinates base at the point $p\in M$ could be obtained in
two ways: either as a coordinate system centered around $p$ such
that $\textrm{Sym}(\partial_{I}\Gamma_{ij}^{k})|_{p}=0\,, |I|\geq
0$, where $\textrm{Sym}(\partial_{I}\Gamma_{ij}^{k})$ means the
symmetrization of $\partial_{I}\Gamma_{ij}^{k}$ with respect to the
sub-indices $I\cup\{i,j\}$; or as linear coordinates on the tangent
vector space $T_{p}M$ defined by the exponential map $\exp_{p}:
T_{p}M\rightarrow M$. Using the second view point, we get the
following description: suppose a point $q$ in a small neighborhood
of $p$ on $M$ is on the geodesic $\gamma(s)=\exp_{p}(sv)$, where
$|v|=1$, and $s$ is the arc-length parameter. Assume
$q=\exp_{p}(sv)$ for some $s$ and fix a coordinate system
$x=\{x^{i}\}$ near $p$ on $M$, then the Gaussian normal coordinates
$\xi=\{\xi^{i}\}$ of $q=\exp_{p}(sv)$ are related to the coordinates
$x=\{x^{i}\}$ by using the equations for the geodesic:
\begin{equation}\label{Gaussiannormalcoordinates}
x^{i}(\exp_{p}(sv))=x^{i}(p)+s\xi^{i}-\sum_{N=2}^{\infty}{1\over
N!}\Gamma^{i}_{N}|_{p}s^{N}\xi^{N}\,,
\end{equation}
where
$\Gamma^{i}_{N}:=\nabla_{N-\{i_{1},i_{2}\}}\Gamma^{i}_{i_{1}i_{2}}$
are computed in $x=\{x^{i}\}$ coordinates, and $N$ is a multi-index
as before.
\subsubsection*{Holomorphic exponential map and canonical coordinates on K\"ahler manifolds}
Now assume $M$ is a K\"ahler manifold whose K\"ahler potential is
$K(z,\bar{z})$, where $z=\{z^{i}\}_{i=1,2\cdots \dim M}$ is a
complex coordinate system. Suppose the base point $p$ is taken to be
$(z_{*},\bar{z}_{*})$. From the second equation in
(\ref{canonicalcoordinates}), one can solve \cite{Higashijima:2001,
Higashijima:2002, Gerasimov:2004yx} for $t$ and get the following
expression similar to (\ref{Gaussiannormalcoordinates}):
\begin{equation}\label{solofcanonicalcoordinates}
t^{i}(z;z_{*},\bar{z}_{*})=K^{i\bar{\jmath}}(
z_{*},\bar{z}_{*};z_{*},\bar{z}_{*})(K_{\bar{\jmath}}
(z,\bar{z}_{*};z_{*},\bar{z}_{*})-K_{\bar{\jmath}}(z_{*},\bar{z}_{*};z_{*},\bar{z}_{*}))\,,
\end{equation}
where a function $f$ defined near the base point
$(z_{*},\bar{z}_{*})$ is denoted by
$f(z,\bar{z};z_{*},\bar{z}_{*})$. The holomorphic function
$f(z,\bar{z}_{*};z_{*},\bar{z}_{*})$ means the degree zero part in
the Taylor expansion of the function
$f(z,\bar{z};z_{*},\bar{z}_{*})$ in $\bar{z}$ centered at
$\bar{z}_{*}$, where one thinks of $(z,\bar{z})$ as independent
coordinates. This will be explained below using holomorphic
exponential map.

The canonical coordinates can not be defined in terms of geodesics in
the Riemannian geometry since the exponential map is in general not
holomorphic. However, there is \cite{Kapranov:1999} a nice
construction of holomorphic exponential map which gives rise to
these canonical coordinates. To define the holomorphic exponential
map, we first regard the complex manifold $M$ as a Riemannian
manifold and thus get the map $\exp^{\mathbb{R}}_{p}:
T^{\mathbb{R}}_{p}M\rightarrow M$. This also defines the Gaussian
normal coordinates $\xi$. Thinking of $T^{\mathbb{R}}_{p}M$ as a
complex vector space equipped with the complex structure induced by
the complex structure on $M$, then in general the map
$\exp_{p}^{\mathbb{R}}: (\xi,\bar{\xi})\mapsto
(z(\xi,\bar{\xi}),\bar{z}(\xi,\bar{\xi}))$ is not holomorphic. Now
with the assumption that the metric $G_{i\bar{\jmath}}(z,\bar{z})$
is analytic in $z,\bar{z}$, we can analytically continue the map
$\exp_{p}^{\mathbb{R}}$ to the corresponding complexifications
$T^{\mathbb{C}}_{p}M, M_{C}=M\times \overline{M}$, where
$\overline{M}$ is the complex manifold with opposite complex
structure as $M$.

The coordinates on the complexifications $T^{\mathbb{C}}_{p}M,
M_{C}=M\times \overline{M}$ are given by $(\xi,\eta)$ and $(z,w)$
respectively, they are the analytic continuation of the coordinates
$(\xi,\bar{\xi}), (z,\bar{z})$ from
$T^{\mathbb{R}}_{p}M\hookrightarrow T^{\mathbb{C}}_{p}M,
\Delta:M\hookrightarrow M_{C}=M\times \overline{M}$ respectively,
where $\Delta: M\rightarrow M\times \overline{M}, p\mapsto
(p,\bar{p})$ is the diagonal embedding. Here the underlying point of
$\bar{p}$ is really the same as $p$, but we have used the barred
notation to indicated that it is a point on the complex manifold
$\overline{M}$.

Since the Christoffel symbols $\Gamma_{ij}^{k}(z,\bar{z})$ are
analytic in $(z,\bar{z})$, we know that the map
$\exp_{p}^{\mathbb{C}}: (\xi,\eta)\mapsto (z(\xi,\eta),w(\xi,\eta))$
is analytic, that is, holomorphic in $(\xi,\eta)$. Moreover, the map
$\exp_{p}^{\mathbb{C}}$ defines a local bi-holomorphism from a small
neighborhood around the point $0$ inside $T_{p}^{\mathbb{C}}M$ to a
small neighborhood of the point $(p,\bar{p})$ inside
$M_{\mathbb{C}}$. One claims that
$\exp^{\mathbb{C}}_{p}|_{T^{1,0}M}$ gives a holomorphic map
$T^{1,0}_{p}M\rightarrow M$ which is a local bi-holomorphism from a
small neighborhood of $0\in T_{p}^{1,0}M$ to a small neighborhood of $p\in
M$. To show that it maps $T^{1,0}_{p}M$ to $M$, we only need to show
that $w\circ \exp^{hol}_{p}|_{T^{1,0}_{p}M}=w(\bar{p})$, that is,
$w(\xi,\eta)|_{\eta=0}=w(\bar{p})$. Recall that $\bar{z}$ and thus
$w$ satisfies the equation for the geodesic equation
\begin{equation*}
{d^{2}\over
ds^{2}}\bar{z}^{k}+\Gamma^{\bar{k}}_{\bar{i}\bar{j}}{d\bar{z}^{\bar{i}}\over
ds}{d\bar{z}^{\bar{j}} \over ds}=0,\, {d\bar{z}^{\bar{k}}\over
ds}(0)=\bar{\xi}^{\bar{k}}=0,\, \bar{z}(0)=\bar{z}(\bar{p})\,.
\end{equation*}
It is easy to see that $w(s)=w(\bar{p})$ is one and thus the unique
solution to the differential equation. Therefore, $w\circ
\exp^{\mathbb{C}}_{p}(\xi,\eta=0)=w(\bar{p})$ as desired. Since
$z(\xi,\eta)$ is holomorphic in both $\xi,\eta$, we know
$z(\xi,\eta=0)$ is holomorphic in $\xi$. The same reasoning for the
exponential map $\exp_{p}^{\mathbb{R}}$ shows that it is locally a
bi-holomorphism.

Hence one gets a holomorphic exponential map
$\exp^{\textrm{hol}}_{p}: T^{1,0}_{p}M\rightarrow M$. We now denote
the coordinate $\xi$ on $T^{1,0}_{p}M$ by $t$, this is then the
canonical coordinates desired since the equation satisfied by $t$
which is similar to (\ref{Gaussiannormalcoordinates}) implies the
second equation in ({\ref{canonicalcoordinates}}). This can be
checked by direct computations.

The exponential maps $\exp^{\mathbb{R}}_{p}$ and $\exp^{hol}_{p}$
are contrasted as follows:
\begin{eqnarray*}
\exp_{p}^{\mathbb{R}}&=&\exp^{\mathbb{C}}_{p}|_{T^{\mathbb{R}}_{p}M}=\exp^{\mathbb{C}}_{p}|_{T^{1,0}_{p}M\oplus
\overline{T_{p}^{1,0}M}}\,,\\
\exp_{p}^{\textrm{hol}}&=&\exp^{\mathbb{C}}_{p}|_{T^{1,0}_{p}M}=\exp^{\mathbb{C}}_{p}|_{j(T^{1,0}_{p}M)=T^{1,0}_{p}M\oplus
\{0\}}\,.
\end{eqnarray*}
where $T^{1,0}_{p}M\oplus \overline{T^{1,0}_{p}M}$ means the image
of the map $ T^{1,0}_{p}M\rightarrow T^{1,0}_{p}M\oplus
T^{0,1}_{p}M,\, v\mapsto (v,v^{*})$, where $v^{*}$ is the complex
conjugate of $v$;
 and $j(T^{1,0}_{p}M)$ is the
image of the map $j: T^{1,0}_{p}M\mapsto T^{1,0}_{p}M\oplus
T^{0,1}_{p}M,\, v\mapsto (v,0)$.

\subsubsection*{Holomorphic limit}

The holomorphic limit of any function $f(z,\bar{z})$ based at
$z_{*}$ is defined as follows. First one analytically continues the
map $f$ to a map defined on $M_{\mathbb{C}}$. Using the fact that
$\exp^{\mathbb{C}}_{p}$ is a local diffeomorphism from
$T_{p}^{\mathbb{C}}M$ to $M_{\mathbb{C}}$, we get $\hat{f}=f\circ
\exp_{p}^{\mathbb{C}}: T_{p}^{\mathbb{C}}M\rightarrow \mathbb{C}$.
The holomorphic limit of $f(z,\bar{z})$ is given by
$\hat{f}|_{j(T^{1,0})}:T^{1,0}_{p}M\rightarrow
T_{p}^{\mathbb{C}}M\rightarrow \mathbb{C}$.

From now on, to maintain consistency with the notations used in the
literature, we shall use $(z,\bar{z}),(t,\bar{t})$ for
$(z,w),(\xi,\eta)$ when considering holomorphic limits, if no
confusion arises. In the following, sometimes we shall drop the
notations $z_{*},\bar{z}_{*}$ for the base point if it is clear from
the context.

\begin{rem}
In the canonical coordinates $t$ on the K\"ahler manifold $M$, the
holomorphic limit is described by $ f\circ
\exp^{\textrm{hol}}=\hat{f}|_{j(T^{1,0})}: T^{1,0}\times
\{0\}\rightarrow \mathbb{C}, t\mapsto f\circ
\exp^{\textrm{hol}}(t)$. In terms of an arbitrary local coordinate
system $z$ on $M$, taking the holomorphic limit of the a function
$f(z,\bar{z})$ at the base point $z_{*}$ is the same as keeping the
degree zero part of the Taylor expansion of $f(z,\bar{z})$ with
respect to $\bar{z}$, where the center of the Taylor expansion is
$\bar{z}_{*}$. That is, it is the evaluation map $ev_{\bar{z}_{*}}:
f(\bullet,\bullet)\mapsto f(\bullet, \bar{z}_{*})$. This is the
limit that is used in the study of topological string theory in
\cite{Bershadsky:1993ta,Bershadsky:1993cx}.

 One thing that needs to
be taken extra care of is the holomorphic limit of $\det G$
appearing in computing the topological string partition functions.
One has
$G_{z^{i}\bar{z}^{\bar{\jmath}}}=G_{t^{a}\bar{t}^{\bar{b}}}{\partial
t^{a}\over
\partial z^{i}}{\partial \bar{t}^{\bar{b}}\over \partial \bar{z}^{\bar{\jmath}}},\,i,\bar{\jmath},a,\bar{b}=1,2,\cdots \dim M$ and $\log \det
G_{z^{i}\bar{z}^{\bar{\jmath}}}=\log \det
G_{t^{a}\bar{t}^{\bar{b}}}+\log \det{\partial t^{a}\over
\partial z^{i}}+\log \det{\partial \bar{t}^{\bar{b}}\over \partial \bar{z}^{\bar{\jmath}}} $.
Since only the holomorphic derivative of $\log \det
G_{z^{i}\bar{z}^{\bar{\jmath}}}$ will appear in the topological
string partition functions (and also in the ring
$\widehat{\mathcal{R}}$ we shall construct below), the purely
anti-holomorphic term will disappear. Moreover, from
(\ref{canonicalcoordinates}) one can see that $\log \det
G_{t^{a}\bar{t}^{\bar{b}}}(t,\bar{t})=\log \det
G_{t^{a}\bar{t}^{\bar{b}}}(t_{*},\bar{t}_{*})$ is independent of
$t$. Therefore, when computing $\log \det
G_{z^{i}\bar{z}^{\bar{\jmath}}}$ one can effectively extract the
purely anti-holomorphic term and the term $\log \det
G_{t^{a}\bar{t}^{\bar{b}}}(t,\bar{t})$, then one only needs to take
the holomorphic limit of the term $\log \det{\partial t^{a}\over
\partial z^{i}}$.
This could also be seen from (\ref{solofcanonicalcoordinates}),
which implies that
\begin{equation} {\partial t^{i}\over \partial
z^{k}}(z,\bar{z}_{*})=K^{i\bar{\jmath}}(z_{*},\bar{z}_{*})
K_{k\bar{\jmath}}(z,\bar{z}_{*})\,.
\end{equation}
Therefore, in the coordinate system $z$, the holomorphic limit of
the metric $G_{k\bar{\jmath}}$, denoted by $\lim G_{k\bar{\jmath}}$,
is given by
\begin{equation}\lim
G_{k\bar{\jmath}}(z,\bar{z})=G_{k\bar{\jmath}}(z,\bar{z}_{*})={\partial
t^{i}\over \partial z^{k}}(z)G_{i\bar{\jmath}}(z_{*},\bar{z}_{*})\,.
\end{equation}
\end{rem}

\subsubsection*{Variation of the holomorphic exponential map and canonical coordinates}

The holomorphic exponential map $\exp^{\textrm{hol}}_{p}$ does not
depend holomorphically on the base point $z_{*}$
\cite{Kapranov:1999}. The canonical coordinates thus also have
non-holomorphic dependence, as we shall also see below in some
examples. This is due to the fact that the space $T_{z_{*}}^{1,0}M$
changes non-holomorphically when $z_{*}$ moves in $M$: that is, $
{\partial\over \partial \bar{z_{*}}}\pi_{J_{z_{*}}}\neq 0\,, $ where
$\pi_{J_{z_{*}}}={1\over 2}(I-iJ_{z_{*}})$ is the projection from
$T^{\mathbb{C}}_{z_{*}}M$ to $T^{1,0}_{z_{*}}M$. For a more precise
discussion on this, see \cite{Kapranov:1999}.

Take $M$ to be the base $\mathcal{M}$ of the Calabi-Yau threefold
family $\pi:\mathcal{X}\rightarrow \mathcal{M}$ and think of
$T_{z_{*}}^{1,0}\mathcal{M}$ as a Lagrangian in
$T^{\mathbb{C}}_{z_{*}}\mathcal{M}$, this then fits in the frame
work of geometric quantization and is related to the basepoint
independence of the total free energy
$\mathcal{Z}=\sum_{g=0}^{\infty}\lambda^{2g-2}\mathcal{F}^{(g)}$ of
the topological string theory for the family, as studied in
\cite{Witten:1993ed}. The background (base point) independence of
$\mathcal{Z}$ tells that it satisfies some wave-like equations on
$\mathcal{M}$ arising from geometric quantization. These equations
are shown \cite{Witten:1993ed} to be equivalent to the master
anomaly equations for $\mathcal{Z}$ in \cite{Bershadsky:1993cx}
which are identical to the holomorphic anomaly equations for the
topological string partition functions $\mathcal{F}^{(g)}$.

\subsection{Examples of canonical coordinates}
In this section we shall compute the canonical coordinates for
some K\"ahler manifolds.
\begin{expl}[Fubini-Study metric]
Consider the Fubini-Study metric defined on $\mathbb{P}^{1}$
\begin{equation*}
\omega_{FS}={i\over
2}{1\over (1+|z|^{2})^{2}}dz\wedge d\bar{z}
\end{equation*}
with K\"ahler potential $K=\ln (1+|z|^{2})$. It follows then
\begin{equation*}
K_{z}={\bar{z}\over (1+|z|^{2})},\quad K_{z\bar{z}}={1\over
(1+|z|^{2})^{2}}, \quad\partial_{z}^{N}K_{\bar{z}}={(-1)^{N+1}N!
\bar{z}^{N-1}\over (1+|z|^{2})^{N+1}},\,N\geq 1\,.
\end{equation*}
At the point $p$ represented by $z_{*}=0$, we can see that
$\partial_{z}^{N}K|_{p}=0=\partial_{z}^{N}K_{z\bar{z}}|_{p}, N\geq
1$. Hence $z$ is the canonical coordinate based at $z_{*}=0$. To
find the canonical coordinate at a generic point $p$ represented by
$z_{*}$, we apply (\ref{solofcanonicalcoordinates}) and get
\begin{equation*}
t(z;z_{*},\bar{z}_{*})=(1+|z_{*}|^{2})^{2}\left({z\over
(1+z\bar{z}_{*})}-{z_{*}\over (1+z_{*}\bar{z}_{*})}\right)\,.
\end{equation*}
In particular, at $z_{*}=0$, this coincides with $z$. The
non-holomorphic dependence on the base point can be easily seen from
this formula.
\end{expl}

\begin{expl}[Poincare metric] Consider the $\textrm{SL}(2,\mathbb{Z})$
invariant metric
\begin{equation*}
\omega={i\over 2} K_{\tau\bar{\tau}}d\tau\wedge d\bar{\tau}={1\over
y^{2}} dx\wedge dy
\end{equation*}
on the Poincare upper half plane
$\mathcal{H}$, where $e^{-K}={\tau-\bar{\tau}\over i}$, $\tau=x+iy$.
Straightforward computations show that
\begin{equation*}
K_{\bar{\tau}}={1\over \tau-\bar{\tau}}, \,
K_{\tau\bar{\tau}}=-{1\over (\tau-\bar{\tau})^{2}}\,.
\end{equation*}
It follows that the canonical coordinate based at $p$ given by
$\tau_{*}$ is
\begin{equation*}
t(\tau;\tau_{*},\bar{\tau}_{*})=-(\tau_{*}-\bar{\tau}_{*})^{2}\left(
{1\over \tau-\bar{\tau_{*}}}-{1\over \tau_{*}-\bar{\tau}_{*}}\right)
\end{equation*}
In particular, if one takes the base point $\tau_{*}=i\infty$, then
the canonical coordinate $t$ coincides with the complex coordinate
on $\mathcal{H}$ from the embedding $\mathcal{H}\hookrightarrow
\mathbb{C}$.
\end{expl}

\begin{expl}[Weil-Petersson metric for elliptic curve
family]\label{WPforellipticcurve} Taking the elliptic curves
parametrized by $\mathcal{H}$. As in the proof of Prop.
\ref{Yukawaintau}, take the holomorphic top form
$\Omega_{\tau}=dz_{\tau}$ on $T_{\tau}$. Using the diffeomorphism
from the fiber $T_{\tau}$ to the fiber $T_{\tau_{*}}$
\begin{equation*}
z_{\tau}={\tau-\bar{\tau}_{*}\over
\tau_{*}-\bar{\tau}_{*}}z_{\tau_{*}}+{\tau_{*}-\tau\over
\tau_{*}-\bar{\tau}_{*}}\bar{z}_{\tau_{*}}\,,
\end{equation*}
one can compute the K\"ahler potential for the Weil-Peterson metric
from
\begin{equation*}
e^{-K(\tau,\bar{\tau};\tau_{*},\bar{\tau}_{*})}=i\int_{T_{\tau}}\Omega_{\tau}\wedge
\overline{\Omega}_{\tau}={\tau-\bar{\tau}\over i}.
\end{equation*}
This is then the Poincare metric on the upper half plane considered
in the above example.
\end{expl}
\begin{expl}
Suppose on the K\"ahler manifold $M$ there exists complex
coordinates $z=\{z^{i}\}$ and a holomorphic function $F(z)$, so that
the K\"ahler metric is given by
\begin{equation*}
\omega={i\over 2}\mathrm{Im}\tau~ dz\wedge d\bar{z}
=i\partial\bar{\partial}K\,,
\end{equation*}
where $K={1\over 2}~\mathrm{Im} w\bar{z},
w_{i}(z)={\partial_{z^{i}}F(z)},\,
\tau_{ij}(z)=\partial_{z^{i}}\partial_{z^{j}}F(z)$. Manifolds
satisfying these properties are studied in detail in
\cite{Freed:1999sm}. The canonical coordinates are then given by
\begin{equation*}
t^{i}(z;z_{*},\bar{z}_{*})={1\over
\tau_{ij}(z_{*})-\bar{\tau}_{ij}(\bar{z}_{*})}(w_{j}(z;z_{*},\bar{z}_{*})-w_{j}(z_{*};z_{*},\bar{z}_{*})-
\bar{\tau}_{jk}(\bar{z}_{*})(z^{k}-z^{k}_{*}))\,.
\end{equation*}
\end{expl}

\subsection{Special K\"ahler metric on deformation spaces} 

Now we take $M$ to be the base of the family $\pi:
\mathcal{X}\rightarrow \mathcal{M}$ of Calabi-Yau threefolds $X$.
Assume that $\dim \mathcal{M}=h(=h^{2,1}(X))$.

Fixing a section $\Omega(z)$ of the the Hodge line bundle
$\mathcal{L}\rightarrow \mathcal{M}$ and choosing a symplectic basis
$\{A^{I},B_{J}\}_{I,J=0,1,\cdots h}$ for $H_{3}(X,\mathbb{Z})$, then
the periods are given by
\begin{equation*}
(\int_{A^{0}}\Omega, \int_{A^{a}}\Omega,
\int_{B^{a}}\Omega,\int_{B_{0}}\Omega)=
(X^{0},X^{a},\mathcal{F}_{a},\mathcal{F}_{0})=X^{0}(1,t^{a},F_{t^{a}},
2F-t^{a}F_{t^{a}})\,,
\end{equation*}
 where $a=1,2,\cdots h$ and $\mathcal{F}(X^{I})$ is \cite{Strominger:1990pd, Bershadsky:1993cx}
 a holomorphic homogeneous function of $X$
of degree $2$. Here the function $F$ is defined by
$({X^{0}})^{-2}\mathcal{F}$ and the sub-indices mean derivatives
with respect to corresponding coordinates.

Now assume that $z_{*}$ is the large complex structure limit defined
by $z=0$ and $A^{0}$ is the vanishing cycle at this point. Then near
the base point $z_{*}$, the quantities
$t^{a}(z;z_{*},\bar{z}_{*})=X^{a}(z;z_{*},\bar{z}_{*})/X^{0}(z;z_{*},\bar{z}_{*})\sim
\ln z^{a}+\mathcal{O}(z^{0}),a=1,2,\cdots h$ gives a local
coordinate system on the manifold $\mathcal{M}$ due to local Torelli
theorem which says that the period map $\mathcal{P}:
\mathcal{M}\rightarrow \mathbb{P}H^{3}(X,\mathbb{C}),\, z\mapsto
[X^{I}(z),F_{J}(z)]$ is a local isomorphism. These coordinates, as
ratios of the periods, are called special coordinates in the
literature. Then the K\"ahler potential of the Weil-Petersson metric
is determined from
\begin{equation}
e^{-K}=iX^{0}\overline{X^{0}}\left(2\overline{F(t)}-2F(t)+(t^{a}-\bar{t}^{a})(F_{a}+\overline{F_{a}})\right),
\end{equation}
Using the fact that the prepotential $F(t)$ has the form
$F(t)={\kappa_{abc}\over 6!}t^{a}t^{b}t^{c}+Q(t)+\sum_{d}N_{d}e
^{dt}$, where $Q(t)$ is a quadratic polynomial of $\{t^{a}\}$, it
can be shown that the special coordinates
$t^{a}(z;z_{*},\bar{z}_{*}),a=1,2,\cdots h$ defined near the large
complex structure limit $z_{*}$ are the canonical coordinates based
at $z_{*}$, see \cite{Bershadsky:1993cx} for details. Moreover,
rewriting the above equation as
\begin{equation}\label{splitting}
e^{-K(z,\bar{z})}=X^{0}\bar{X^{0}}e^{-\mathrm{K}(t,\bar{t})},\quad
e^{-\mathrm{K}(t,\bar{t})}=i
\left(2\overline{F(t)}-2F(t)+(t^{a}-\bar{t}^{a})(F_{a}+\overline{F_{a}})\right)\,,
\end{equation}
one then gets \cite{Ferrara:1991aj}
\begin{equation}
K_{z^{i}}=-\partial_{z^{i}}\log X^{0}+\mathrm{K}_{t^{a}}{\partial
t^{a}\over
\partial z^{i}},\quad \Gamma_{z^{i}z^{j}}^{z^{k}}={\partial z^{k}\over \partial t^{a}}{\partial \over \partial z^{i}} {\partial t^{a}\over
\partial z^{j}}+{\partial z^{k}\over \partial t^{c}}\mathrm{\Gamma}_{t^{a}t^{b}}^{t^{c}}{\partial t^{a}\over \partial z^{j}}
{\partial t^{b}\over \partial z^{j}}\,,
\end{equation}
where $\mathrm{\Gamma}_{t^{a}t^{b}}^{t^{c}}$ is computed in the
metric given by the K\"ahler potential $\mathrm{K}(t,\bar{t})$. Then
at the large complex structure limit $z_{*}$, since $\{t^{a}\}$ are
the canonical coordinates, according to
(\ref{canonicalcoordinates}), one has the following holomorphic
limits:
\begin{equation}\label{hollimitsofconnections}
\lim K_{z^{i}}=-\partial_{z^{i}}\log X^{0},\quad \lim
\Gamma_{z^{i}z^{j}}^{z^{k}}={\partial z^{k}\over \partial
t^{a}}{\partial \over \partial z^{i}} {\partial t^{a}\over
\partial z^{j}}\,.
\end{equation}
In the remaining of this work we will only consider the holomorphic
limit based at the large complex structure $z_{*}=0$ which is given
by $\bar{t}=\overline{ i\infty}$, and simply denote this limit by
{\it lim} without specifying the base point. This limit is
interesting since it is in this particular limit that the
topological string partition functions on a Calabi-Yau threefold
$X$ are identical (under the mirror map) to the generating functions
of Gromov-Witten invariants of its mirror manifold $\check{X}$.

\subsection{Ring of Yamaguchi-Yau and the construction of the triple}

In this section, we shall construct the ring
$\widehat{\mathcal{R}}$. We shall review the construction of a ring
in \cite{Yamaguchi:2004bt} by Yamaguchi-Yau for the quintic mirror
family. The purpose is to reduce the number of generators for the
algebra $\widetilde{\mathcal{R}}$ defined above and also find its
non-holomorphic completion $\widehat{\mathcal{R}}$.

The construction of Yamaguchi-Yau says that the antiholomorphic
dependence of the normalized topological string partition functions
$F^{(g)}=(X^{0})^{2g-2}\mathcal{F}^{(g)}$ are encoded in the
generators
\begin{equation*}
\theta^{i}\log e^{-K},i=1,2,3,\quad \theta \log \det G\,,
\end{equation*}
while the coefficients are polynomials of
\begin{equation*}
\theta\log\alpha^{3}C_{\alpha\alpha\alpha}=\theta\log{\kappa\over\beta}={\alpha\over
\beta}\,.
\end{equation*}
More precisely, according to the Picard-Fuchs equation
(\ref{PFforcompactCY}) and the definition (\ref{Kahlerpotential}),
one has
\begin{equation}
\mathcal{L}~e^{-K}=\left(\theta^{4}-\alpha\prod_{i=1}^{4}
(\theta+c_{i})\right)e^{-K}=0\,,
\end{equation}
where $c_{i}=i/5,i=1,2,3,4$. This then implies that
$\theta^{4}e^{-K}$ is a polynomial of $\theta^{i}e^{-K},i=1,2,3$ and
${\alpha\over \beta}$. The special geometry relation
(\ref{specialgeometryrelation}) implies that
\begin{equation}\label{integratedspecialgeometry}
\partial_{\alpha}\bar{\partial}_{\bar{\alpha}}\Gamma_{\alpha\alpha}^{\alpha}=
\partial_{\alpha}G_{\alpha\bar{\alpha}}-\partial_{\alpha}(e^{2K}G^{\alpha\bar{\alpha}}G^{\alpha\bar{\alpha}}
\overline{C}_{\bar{\alpha}\bar{\alpha}\bar{\alpha}})\,.
\end{equation}
It follows then
\begin{eqnarray*}
&&\bar{\partial}_{\bar{\alpha}}[\partial_{\alpha}
\Gamma_{\alpha\alpha}^{\alpha}+(\Gamma_{\alpha\alpha}^{\alpha})^{2}-2\Gamma_{\alpha\alpha}^{\alpha}\partial_{\alpha}
K\\
&&-4\partial_{\alpha}K_{\alpha}+2(\partial_{\alpha}
K)^2+(\partial_{\alpha} \log
C_{\alpha\alpha\alpha})(2\partial_{\alpha}
K-\Gamma_{\alpha\alpha}^{\alpha}) ]=0\,.
\end{eqnarray*}
Hence we know
\begin{eqnarray*}
&&\partial_{\alpha}
\Gamma_{\alpha\alpha}^{\alpha}+(\Gamma_{\alpha\alpha}^{\alpha})^{2}-2\Gamma_{\alpha\alpha}^{\alpha}\partial_{\alpha}
K\\
&&-4\partial_{\alpha}K_{\alpha}+2(\partial_{\alpha}
K)^2+(\partial_{\alpha} \log
C_{\alpha\alpha\alpha})(2\partial_{\alpha}
K-\Gamma_{\alpha\alpha}^{\alpha}) =f_{\alpha}\,
\end{eqnarray*}
for some holomorphic function $f_{\alpha}$. Taking the holomorphic
limit of the left hand side, according to
(\ref{hollimitsofconnections}), we get
\begin{eqnarray*}
&&\partial_{\alpha}^{2} \log {\partial t\over
\partial \alpha}+(\partial_{\alpha} \log {\partial t\over
\partial \alpha})^{2}+2\partial_{\alpha} \log {\partial t\over
\partial \alpha}\partial_{\alpha}\log X^{0}
+4\partial_{\alpha}\partial_{\alpha}\log X^{0}\\
&&+2(\partial_{\alpha}\log X^{0})^2+(\partial_{\alpha} \log
C_{\alpha\alpha\alpha})(-2\partial_{\alpha} \partial_{\alpha}\log
X^{0}-\partial_{\alpha} \log {\partial t\over
\partial \alpha}) =f_{\alpha}\,.
\end{eqnarray*}
The holomorphic function was fixed in \cite{Lian:1994zv,
Yamaguchi:2004bt} (see also \cite{Hosono:2008ve}) to be ${1-{12\over
5}\alpha\over \alpha^{2}\beta}$.

One can also replace the coordinate $\alpha$ in
(\ref{integratedspecialgeometry}) by $x=\ln \alpha$ defined locally
on the punctured deformation space, then we get
\begin{eqnarray}\label{secondderivativeofgamma}
&&\theta^{2} \log G_{x\bar{x}}+(\theta\log G_{x\bar{x}})^{2}-2\theta
\log G_{x\bar{x}}\theta K\nonumber\\
&& -4\theta^{2}K+2(\theta K)^2+(\theta \log C_{xxx})(2\theta
K-\theta \log G_{x\bar{x}})=f_{x}\,,
\end{eqnarray}
where $\theta=\partial_{x}=\alpha{\partial\over \partial
\alpha},C_{xxx}=\alpha^{3}C_{\alpha\alpha\alpha}={\kappa\over
\beta},\theta\log C_{xxx}={\alpha\over \beta}$ and $f_{x}$ is
another holomorphic function. Now we take the holomorphic limit of
the above identity and get
\begin{eqnarray}\label{secondderivativeofholgamma}
&&\theta^{2} \log \theta t+(\theta\log \theta t)^{2}+2\theta \log
\theta t~\theta \log
X^{0}\nonumber\\
&&4\theta^{2}\log X^{0}+2(\theta \log X^{0})^2+(\theta \log
C_{xxx})(-2\theta \log X^{0}-\theta \log \theta t) =f_{x}\,,
\end{eqnarray}
with
\begin{equation*}
f_{x}={2\over 5}{\alpha\over \beta}\,.
\end{equation*}
Therefore, as shown in \cite{Yamaguchi:2004bt}, one gets the
following Yamaguchi-Yau ring
\begin{equation}
\mathcal{R}_{YY}=\mathbb{C}[\theta^{i}\log e^{-K}, i=1,2,3;
\Gamma_{xx}^{x}=\theta\log G_{x\bar{x}},\theta\log
C_{xxx}={\alpha\over\beta}]\,.
\end{equation}
Note that
\begin{equation}\label{generatorofholamb}
\theta \theta\log C_{xxx}=\theta {\alpha\over \beta}={\alpha\over
\beta^{2}}=\theta \log C_{xxx} (\theta \log C_{xxx}+1)\,,
\end{equation}
then the ring $\mathcal{R}_{YY}$ is closed under taking the
derivative $\theta$. The generators of this ring
$(\mathcal{R}_{YY},\theta)$ are essentially
$K_{x},K_{xx},K_{xx},\Gamma_{xx}^{x},\theta \log C_{xxx}$.

However, it is not convenient to directly interpret this as the
analogue of the ring of almost-holomorphic modular forms. For this
reason, we connect this ring $(\mathcal{R}_{YY},\theta)$ to
$(\widetilde{\mathcal{R}},D)$.

Due to (\ref{secondderivativeofholgamma}), and the relation between
the derivatives $\theta$ and $D$ given by $\theta= \beta^{-1}
(X^{0}\theta t)^{-2}D$, we know that the set of generators for
$\widetilde{\mathcal{R}}$ could be reduced to $D^{i}X^{0},i=0,1,2,3;
D^{j}\theta t,j=0,1; B,C$. Recall that
$\mathcal{R}=\mathbb{C}[(X^{0})^{\pm 1}, (\theta t)^{\pm 1}, B^{\pm
1},C^{\pm 1}]$, then one can see that
\begin{eqnarray*}
&\widetilde{\mathcal{R}}&=\mathbb{C}[D^{i}\log X^{0},i=1,2,3;
D^{j}\log \theta t,j=1; \alpha,\beta]\otimes \mathcal{R}\,.
\end{eqnarray*}
Recall (\ref{splitting}), we get the following
\begin{eqnarray*}
&&\theta\log e^{-K(x,\bar{x})}=\theta \log(
X^{0}\bar{X}^{0}e^{-\mathrm{K}(t,\bar{t})})=\theta\log
X^{0}+\theta\log
e^{-\mathrm{K}(t,\bar{t})}\,\\
&&\theta\log G_{x\bar{x}}=\theta \log( \theta t~\bar{\theta
t}~G_{t\bar{t}})=D\log \theta t+\theta\log G_{t\bar{t}}\,.
\end{eqnarray*}
their holomorphic limits are
\begin{equation*}
\lim ~\theta\log e^{-K(x,\bar{x})}=\theta \log X^{0},\quad \lim~
\theta \log G_{x,\bar{x}}=\theta\log \theta t\,,
\end{equation*}
Therefore, the holomorphic limit of the ring $\mathcal{R}_{YY}$ is
given by
\begin{equation*}
\lim \mathcal{R}_{YY}=\mathbb{C}[\theta^{i}\log X^{0}, i=1,2,3;
\theta\log \theta t,{\alpha\over\beta}=\theta\log C_{xxx}]
\end{equation*}
That is, the generators $D^{i}\log X^{0},i=1,2,3; D\log \theta t$ in
$\widetilde{\mathcal{R}}_{0,0}$ are equivalent to the holomorphic
limits of the non-holomorphic generators in $ \mathcal{R}_{YY}$. It
follows then that
\begin{equation}
\widetilde{\mathcal{R}}=\lim \mathcal{R}_{YY}\otimes
 \mathcal{R}\,.
\end{equation}
This motivates us to define the non-holomorphic completion
$\widehat{\mathcal{R}}$ of $\widetilde{\mathcal{R}}$ as
\begin{equation}
\widehat{\mathcal{R}}=\mathcal{R}_{YY}\otimes  \mathcal{R}\,.
\end{equation}
Moreover, $F^{(g)}\in \mathcal{R}_{YY}\subseteq
\widehat{\mathcal{R}}_{0,0}$, where $\mathcal{R}_{YY}$ and
$\widehat{\mathcal{R}}_{0,0}$ are only differed by the holomorphic
generators of degree and weight zero.

\subsection{Summary of results}

In summary, in section 3 we constructed
$(\widetilde{\mathcal{R}},D)$ as a graded differential ring which
is an analogue of the ring of quasi-modular forms. In the last
subsection we then used special K\"ahler geometry to refine the
generators of the ring to get
\begin{eqnarray*}
\mathcal{R}&=&\mathbb{C}[(X^{0})^{\pm 1},(\theta t)^{\pm 1}, B^{\pm 1},C^{\pm 1}]\,,\\
\widetilde{ \mathcal{R}}&=&\mathcal{R}\otimes
\mathbb{C}[D^{i}\log X^{0},i=1,2,3;  D\log\theta t]\,,\\
\widehat{ \mathcal{R}}&=&\mathcal{R}\otimes \mathbb{C}[D^{i}\log
e^{-K},i=1,2,3; D\log\det G_{x\bar{x}}]\,.
\end{eqnarray*}

Recall the structure of the graded rings
$(M_{*}(\Gamma),\widetilde{M}_{*}(\Gamma),\widehat{M}_{*}(\Gamma))$
defined for $\pi_{\Gamma}: \mathcal{E}_{\Gamma}\rightarrow
X_{\Gamma}$
\begin{eqnarray*}
&& \partial_{\tau}: M_{*}(\Gamma)\rightarrow \widetilde{M}_{*}(\Gamma)\,,\\
 &&\textrm{``modular completion"}
:~~\widetilde{M}_{*}(\Gamma)\rightarrow
\widehat{M}_{*}(\Gamma)\subseteq \widetilde{M}_{*}(\Gamma)[Y], \quad
\quad Y={1\over 12}{-3\over \mathrm{Im}\tau}\,,
\\
&&\textrm{``constant term map"}~~ Y\rightarrow
0:~\widehat{M}_{*}(\Gamma)\rightarrow \widetilde{M}_{*}(\Gamma)
\,,\\
&& \partial_{\tau}: \widetilde{M}_{k}(\Gamma)\rightarrow
\widetilde{M}_{k+2}(\Gamma)\,,\\
&& \hat{\partial_{\tau}}=\partial_{\tau}+kY:
\widehat{M}_{k}(\Gamma)\rightarrow \widehat{M}_{k+2}(\Gamma)\,.
\end{eqnarray*}
From (\ref{splitting}), we know
\begin{eqnarray}\label{nonholgeneratorsplitting}
&&D\log e^{-K(x,\bar{x})}=D \log(
X^{0}\bar{X}^{0}e^{-\mathrm{K}(t,\bar{t})})=D\log X^{0}+D\log
e^{-\mathrm{K}(t,\bar{t})}\\
&&D\log G_{x\bar{x}}=D \log( \theta t~\bar{\theta
t}~G_{t\bar{t}})=D\log \theta t+D\log G_{t\bar{t}}\,.
\end{eqnarray}
Define $Y_{1}=D\log G_{t\bar{t}},Y_{2}=-D\log
e^{-\mathrm{K}(t,\bar{t})}$, then we have the following analogue
between $(\mathcal{R},
\widetilde{\mathcal{R}},\widehat{\mathcal{R}})$ defined for $\pi:
\mathcal{X}\rightarrow \mathcal{M}$ and
$(M_{*}(\Gamma),\widetilde{M}_{*}(\Gamma),\widehat{M}_{*}(\Gamma))$
defined for $\pi_{\Gamma}: \mathcal{E}_{\Gamma}\rightarrow
X_{\Gamma}$:
\begin{eqnarray*}
&& D: \mathcal{R}\rightarrow \widetilde{\mathcal{R}}\,,\\
 &&\textrm{``non-holomorphic completion"}
:~~\widetilde{\mathcal{R}}\rightarrow \widehat{\mathcal{R}}\subseteq
\widetilde{\mathcal{R}}[Y_{1},Y_{2}]\,\\
&&\textrm{``holomorphic limit"}~~ Y_{1},Y_{2}\rightarrow
0:~\widehat{R}\rightarrow \widetilde{R}
\,,\\
&& D: \widetilde{R}_{k,m}(\Gamma)\rightarrow
\widetilde{R}_{k+2,m}(\Gamma)\,,\\
&& \hat{D}=D+kD\log G _{x\bar{x}}+m (-D\log e^{-K}):
\widehat{R}_{k,m}\rightarrow \widehat{R}_{k+2,m}\,.
\end{eqnarray*}
where the operator $\hat{D}$ comes from the covariant derivative
$\partial_{x}+k\Gamma_{xx}^{x}+mK_{x}$ on sections of
$\text{Sym}^{\otimes k}T\mathcal{M}\otimes \mathcal{L}^{m}$.

The above construction for $(\mathcal{R},
\widetilde{\mathcal{R}},\widehat{\mathcal{R}})$ could also be
formally applied to the elliptic curve families in
(\ref{periodsofcurve}), see \cite{Hosono:2008ve}. The Weil-Petersson
metric is determined from
$e^{-K(\alpha,\bar{\alpha})}=i\omega_{0}\bar{\omega_{0}}(\tau-\bar{\tau})
=i\omega_{0}\bar{\omega_{0}}e^{-\mathrm{K}(\tau,\bar{\tau})}$. The
quantities $Y_{1},Y_{2}$ are now computed to be ${-2 \over
12}{-3\over \pi\textrm{Im}~\tau}$ and ${-1 \over 12}{-3\over
\pi\textrm{Im}~\tau}$, respectively.
The triple
$(\mathcal{R},\widetilde{\mathcal{R}},\widehat{\mathcal{R}})$
coincides with the triple $(M_{*}(\Gamma),\widetilde{M}_{*}(\Gamma),
\widehat{M}_{*}(\Gamma))$, as well as the maps among the members in
the triple.

For the non-compact Calabi-Yau threefold families in
(\ref{PFfornoncompactCY}), one has $X^{0}=1$ and $\theta t=A$. The
rings $(\mathcal{R},\widetilde{\mathcal{R}})$ coincide with
$(M_{*}(\Gamma),\widetilde{M}_{*}(\Gamma))$, as mentioned earlier in
this paper. But the explicit forms for $Y_{1},Y_{2}$ are difficult
to compute in these cases. \footnote{This is because the
Picard-Fuchs equation for a non-compact Calabi-Yau threefold family
has only three periods, and the K\"ahler potential of the
Weil-Petersson metric cannot be computed as the compact cases. One
needs to compactify \cite{Chiang:1999tz} the non-compact Calabi-Yau
threefold to a compact Calabi-Yau geometry, and then do
computations there, after that one takes the decompactification
limit of corresponding quantities.}

It is easy to see that one should be able to apply the same
construction for the quintic mirror family to construct triples
$(\mathcal{R},\widetilde{\mathcal{R}},\widehat{\mathcal{R}})$ for
other one-parameter Calabi-Yau threefold families whose
Picard-Fuchs equation takes the form as
(\ref{PFforcompactCYgeneral}) with $\sum_{i=1}^{4}c_{i}=2$. The only
thing that needs to be checked is that the function $f_{x}$ in
(\ref{secondderivativeofholgamma}) is contained in
$\mathbb{C}[B^{\pm 1},C^{\pm 1}]$. In fact, for many Calabi-Yau
families \cite{Lian:1994zv, Hosono:1996jv, Yamaguchi:2004bt}, this
holomorphic function is a rational function\footnote{The author
thanks Prof. Shinobu Hosono for email correspondences and telling
him the references on this.}. We shall not discuss the details in
this work.

\subsection{Special geometry polynomial ring}

Most of the generators in $\widehat{\mathcal{R}}$ obtained from the
elements in $\mathcal{R}_{YY}$ have weight zero.
 In \cite{Alim:2013}, a set of the non-holomorphic, positive weight generators for
 $\widehat{\mathcal{R}}$
are chosen so that no negative powers of the generators appear upon
taking the derivative $D$. The particular form of the ring
$\widehat{\mathcal{R}}$ is termed the special polynomial ring in
\cite{Alim:2013}. For completeness, in the following we shall review
the construction of the generators therein.

First notice that the set of generators in $\widehat{\mathcal{R}}$
given by $ X^{0}D^{i}\log e^{-K}, i=1,2,3;$ $\theta t~ D\log \det
G_{x\bar{x}} $ is equivalent the set of generators
$S^{xx},S^{x},S,K_{x}$ in (\ref{polyring}). The reason is as
follows. Integrating the special geometry relation
(\ref{specialgeometryrelation}), we then get
\begin{equation}
\Gamma_{xx}^{x}=2K_{x}-C_{xxx}S^{xx}+s_{xx}^{x}
\end{equation}
then up to multiplication and addition by $K_{x}$ and holomorphic
quantities, $S^{xx}$ is essentially $\Gamma_{xx}^{x}=\theta \log
\det G_{x\bar{x}}$. The first or last equation in (\ref{polyring})
tells that $S^{x}$ is essentially $\partial_{x}K_{x}$, and the
seconde tells that $S$ is $\partial_{x}^{2}K_{x}$. Moreover, the
derivatives of the generators in $\mathcal{R}_{YY}$ coincide with
those for the generators $S^{xx},S^{x},S,K_{x}$ in
(\ref{polyring}). \\

Now a nice set of generators for the special geometry polynomial
ring $\widehat{\mathcal{R}}$ can be chosen as follows. First one
makes the following change of generators \cite{Alim:2007qj}
\begin{equation*}\label{shift}
\tilde{S}^{tt} = S^{tt},\, \tilde{S}^t = S^t - S^{tt} K_t,\,
\tilde{S} = S- S^t K_t + \frac{1}{2} S^{tt} K_t K_t ,\, \tilde{K}_t=
K_t
\end{equation*}
Then as before one defines $ \tau=\frac{1}{2\pi i}\kappa^{-1}
\partial_{t} F_{t} $ which gives $\frac{\partial \tau}{\partial
t}=\frac{1}{2\pi i}\kappa^{-1}C_{ttt}$.
 Then one forms the following quantities on the
deformation space $\mathcal{M}$:
  \begin{equation}
    \begin{aligned}
      K_0 &=  \kappa C_{ttt}^{-1} \, (\theta t)^{-3} \,,&  G_1&=\theta t\,,& K_2&=\kappa C_{ttt}^{-1}\tilde{K}_t\,,\nonumber\\
      T_2&=\tilde{S}^{tt}\,,& T_4&= C_{ttt}^{-1} \tilde{S}^t\,,&
      T_6&=C_{ttt}^{-2} \tilde{S}\,,
    \end{aligned}
\end{equation}
where the propagators $\tilde{S}^{tt},\tilde{S}^{t},\tilde{S}$ are
normalized by suitable powers of $X^{0}$ so that they are sections
of $\mathcal{L}^{0}$. That is, they have degree zero. The weights of
these generators are the sub-indices they carry. It follows that the
derivatives of the generators of $\widehat{\mathcal{R}}$ given in
(\ref{polyring}) now become ($\partial_{\tau}:={1\over 2\pi
i}{\partial\over \partial \tau}$)
\begin{equation}
\begin{aligned}
\partial_{\tau}K_0&=-2K_0\,K_2- K_0^2\, G_1^2\,(\tilde{h}^\alpha_{\alpha\alpha\alpha}+3(s_{\alpha\alpha}^\alpha+1))\,,\\
\partial_{\tau} G_1&= 2G_1\,K_2-  \kappa G_1\,T_2\,+K_0 G_1^3(s_{\alpha\alpha}^\alpha+1)\,,\\
\partial_{\tau} K_2&=3K_2^2-3 \kappa K_2\,T_2- \kappa^{2}T_4+K_0^2\,G_1^4 k_{\alpha\alpha}-K_0\,G_1^2\,K_2\,\tilde{h}^\alpha_{\alpha\alpha\alpha}\,,\\
\partial_{\tau} T_2&=2K_2\,T_2- \kappa T_2^2+2 \kappa T_4+\kappa^{-1}K_0^2 G_1^4 \tilde{h}^\alpha_{\alpha\alpha}\,,\\
\partial_{\tau} T_4&=4 K_2 T_4-3 \kappa T_2\,T_4+ 2\kappa  T_6-K_0\, G_1^2 \, T_4 \tilde{h}^\alpha_{\alpha\alpha\alpha}
- \kappa^{-1}K_0^2\, G_1^4 \,T_2 k_{\alpha\alpha}+\kappa^{-2}K_0^3\, G_1^6 \tilde{h}_{\alpha\alpha}\,,\\
\partial_{\tau} T_6&= 6 K_2\, T_6- 6 \kappa T_2 \,T_6+\frac{\kappa}{2} T_4^2- \kappa^{-1}K_0^2\, G_1^4 \,T_4\,k_{\alpha\alpha}
+\kappa^{-3}K_0^4\, G_1^8 \tilde{h}_\alpha-2 \, K_0\,G_1^2\,T_6
\tilde{h}^\alpha_{\alpha\alpha\alpha}\,.
\end{aligned}
\end{equation}
The quantities
$\tilde{h}^\alpha_{\alpha\alpha\alpha},s_{\alpha\alpha}^\alpha
,k_{\alpha\alpha},\tilde{h}^\alpha_{\alpha\alpha\alpha},\tilde{h}^\alpha_{\alpha\alpha}
,\tilde{h}_{\alpha\alpha},\tilde{h}_\alpha$ are holomorphic
functions. It turns out that they are polynomials of an additional
generator $C_{0}=\theta \log C_{xxx}={\alpha\over\beta}$ with
\begin{equation}
\partial_{\tau} C_{0}=C_{0}(C_{0}+1)G_{1}^{2}\,.
\end{equation}
These explicit polynomials for the quintic mirror family could be
found in \cite{Alim:2013} and are omitted here.

\subsection{Holomorphic anomaly equations}

As mentioned earlier in section 4.4, one has
$F^{(g)}:=(X^{0})^{2g-2}\mathcal{F}^{(g)}\in\mathcal{R}_{YY}\subseteq
\widehat{\mathcal{R}}_{0,0}$.

The holomorphic anomaly equations then become \cite{Alim:2013}
\begin{eqnarray*}
&&\frac{\partial F^{(g)}}{\partial T_2}-{1\over
\kappa}\frac{\partial F^{(g)}}{\partial T_4} \,K_2 + {1\over
\kappa^{2}}\frac{\partial F^{(g)}}{\partial T_6} K_2^2 =\frac{1}{2}
\sum_{r=1}^{g-1}
\partial_t F^{(g-r)}\,\partial_t\, F^{(r)} +\frac{1}{2} \partial_t^2
F^{(g-1)}\,,\\
&&\frac{\partial F^{(g)}}{\partial K_2}=0\,,
\end{eqnarray*}
where $\partial_{t}=(X^{0})^{-2}(C_{0}+1)(\theta
t)^{-3}\partial_{\tau}$.
\begin{expl}
Consider the Calabi-Yau threefold family
 $\pi: \mathcal{X}\rightarrow \mathcal{M}$ which is mirror to the $K_{\mathbb{P}^{2}}$
family (with a one--dimensional base parametrizing the complexified
K\"ahler structures of $K_{\mathbb{P}^{2}}$). It is proved in
\cite{Alim:2013} and also mentioned earlier in section 3.2 that
$\mathcal{M}\cong X_{0}(3), \widetilde{\mathcal{R}}\cong
\widetilde{M}_{*}(\Gamma_{0}(3))=\mathbb{C}[A^{\pm 1},B^{\pm
1},C^{\pm 1},E]$.
In this case one can consistently choose the generators so that
$T_4=T_6=K_2=0,T_2=\frac{\hat{E}}{2}$ with
$\partial_{t}=\kappa^{-1}C_{ttt}\partial_{\tau}=
B^{-3}\partial_{\tau}$. Then the holomorphic anomaly equations
simplify greatly. In particular, the equation for the holomorphic
limits of $F^{(g)}=(X^{0})^{2g-2}\mathcal{F}^{(g)}$ at the large
complex structure, denoted by $F_{g}\in
\widetilde{\mathcal{R}}_{0,0}\subseteq\mathbb{C}[A^{\pm 1},B^{\pm
1},C^{\pm 1},E]$, becomes
\begin{equation*}
\partial_{E}F_{g}={1\over
4B^{6}}\left(\sum_{r=1}^{g-1}\partial_{\tau}F_{g-r}\partial_{\tau}F_{r}
-{E-A^{2}\over 2}
\partial_{\tau}F_{g-1}+\partial_{\tau}\partial_{\tau}F_{g-1}\right)\,.
\end{equation*}
The boundary conditions \cite{Bershadsky:1993ta, Bershadsky:1993cx}
at the large complex structure $\alpha=0$ limit and the gap
condition \cite{Ghoshal:1995wm, Huang:2006si} at the conifold point
$\alpha=1$ for the topological string partition functions now
translate to the regularity conditions for the quasi-modular form
$F_{g}$ at the two cusp classes $[i\infty],[0]$ on $X_{0}(3)$. The
Fricke involution $W_{N}: \tau\mapsto -{1\over 3\tau}$ translates
further these conditions to some conditions on the $q_{\tau}=\exp
2\pi i\tau$ expansion of $F_{g},F_{g}|_{W_{N}}$ at the infinity cusp
$[i\infty]$. This then allows one to solve $F_{g}$ and thus
$F^{(g)}$ genus by genus recursively. One can also prove
\cite{Zhou:2014thesis} the existence and uniqueness of the solutions to
the holomorphic anomaly equations with the provided boundary
conditions.
\end{expl}

\section{\textup{Conclusions and discussions}}
We constructed the graded rings
$(\mathcal{R},\widetilde{\mathcal{R}},\widehat{\mathcal{R}})$ on the
deformation space $\mathcal{M}$ from the periods of the Picard-Fuchs
equation and special K\"ahler geometry on the deformation space. A
parallelism between these rings and the rings
$M_{*}(\Gamma),\widetilde{M}(\Gamma),\widehat{M}(\Gamma)$ was made:
the way they were constructed; non-holomorphic completion and
modular completion; holomorphic limit and ``constant term map''. We
further showed that in some special cases the rings
$(\mathcal{R},\widetilde{\mathcal{R}})$ are equivalent to the rings
of modular quantities $(M_{*}(\Gamma),\widetilde{M}(\Gamma))$. These
give some evidences that indeed the graded rings
$(\mathcal{R},\widetilde{\mathcal{R}},\widehat{\mathcal{R}})$ are
analogues of the rings of modular objects
$M_{*}(\Gamma),\widetilde{M}(\Gamma),\widehat{M}(\Gamma)$. We also
discussed some of their applications in solving the holomorphic
anomaly equations.
\[
\xymatrix{
\hat{D}\circlearrowright\widehat{M}(\Gamma)\ar@{->}[d]_{Y\rightarrow
0}&&&& \widehat{\mathcal{R}}\ar@{->}[d]_{\textrm{holomorphic~limit}}
\circlearrowleft\hat{D}\\
D\circlearrowright\widetilde{M}(\Gamma)\ar@/_/[u]_{\textrm{modular~completion}}&&&&\widetilde{\mathcal{R}}
\ar@/_/[u]_{\textrm{non-holomorphic~completion}}\circlearrowleft
D\\
 {M}(\Gamma)\ar@{->}[u]_{D}&&&& \mathcal{R}\ar@{->}[u]_{D}\\
\pi_{\Gamma}: \mathcal{E}_{\Gamma} \rightarrow X_{\Gamma}&&&&\pi:
\mathcal{X}\rightarrow \mathcal{M}}
\]
In the above construction of the triple of graded rings
$(\mathcal{R},\widetilde{\mathcal{R}},\widehat{\mathcal{R}})$, the
parameter $\tau={1\over 2\pi i}\kappa^{-1}F_{tt}$ defined in
(\ref{defoftau}) on the deformation space $\mathcal{M}$ was
introduced to match the known modularity in the non-compact
examples. There are a number of interesting questions about this
quantity $\tau$ we would like to address here and wish to pursue in
the future.

\subsubsection*{Variation of Hodge structures}

For the family $\pi:\mathcal{X}\rightarrow \mathcal{M}$ of
non-compact Calabi-Yau threefolds discussed above, the parameter
$\tau$ is exactly the transcendental modulus for elliptic curve
$\mathcal{E}_{\alpha}$ sitting inside the non-compact Calabi-Yau
threefold $\mathcal{X}_{\alpha}$. It is the normalized period for
the elliptic curve and lies in the upper half plane.
 This results from the fact that the vector space of the periods
$(1,t,F_{t})$ of $\mathcal{X}_{\alpha}$ is closed under the
monodromy, and  upon taking derivatives these periods become
$(0,\omega_{0},\omega_{1})$, where the latter two are the two
periods of $\mathcal{E}_{\alpha}$. In other words, while the three
periods $(1,t, F_{t})$ characterizes the variation of complex
structure of the Calabi-Yau threefold, the quantities $(\theta t,
\theta F_{t})$ characterizes the variation of complex structure of
the elliptic curve sitting inside it.

However, for a general one-parameter compact Calabi-Yau threefold
family, e.g., the quintic mirror family, the vector space of periods
$(X^{0},X^{0}t,X^{0}F_{t})$ is not invariant under the monodromy
group. It is not clear what the geometric meaning of $\tau={1\over
2\pi i}\kappa^{-1}F_{tt}$ is.

\subsubsection*{Enumerative content of $\tau$ and integrality
}\footnote{The author thanks Murad Alim, Yaim Cooper and Shing-Tung
Yau for discussions on this.} For the particular non-compact
geometries (\ref{PFfornoncompactCY}), the $F_{g}$s solved
\cite{Alim:2013} from the holomorphic anomaly equations are explicit
quasi-modular functions in $\tau$ (see also \cite{Aganagic:2006wq}
for related work). Whether the $q_{\tau}$ expansions of the
topological string partition functions have any enumerative content
and how the $q_{t}$ and $q_{\tau}$ expansions are related beg an
explanation\footnote{See \cite{Mohri:2000kf, Stienstra:2005wy} and the more recent \cite{Zhou:2014thesis} for related discussions.}.

Now we briefly recall how the partition functions are related to the
generating functions Gromov-Witten variants under the mirror
symmetry conjecture.
 For the Calabi-Yau threefold family $\pi:
\mathcal{X}\rightarrow \mathcal{M}$ whose generic fiber is $X$,
suppose the mirror family is given by  $\check{\pi}:
\check{\mathcal{X}}\rightarrow \check{\mathcal{M}}$ whose generic
fiber is $\check{X}$. Mirror symmetry predicts the holomorphic limit
$F_{g}=\lim (X^{0})^{2g-2}\mathcal{F}^{(g)}$ at the large complex
structure limit is identical to the generating function of genus $g$
Gromov-Witten invariants of $\check{X}$, that is,
\begin{equation}\label{Fgqt}
F_{g}(t)=\sum_{d=0}^{\infty}N^{\textrm{GW}}_{g,d}q_{t}^{d},\quad
q_{t}=e^{t}\,.
\end{equation}
Recall the prepotential $F(t)$ is given by
\begin{equation*}
F(t)={\kappa\over
3!}t^{3}+\sum_{d=1}^{\infty}N^{\textrm{GW}}_{g=0,d}q_{t}^{d}\,,
\end{equation*}
then $\tau={1\over 2\pi i}\kappa^{-1}F_{tt}$ is the function
determined from
\begin{equation}\label{tauq}
2\pi i \tau= t+\kappa^{-1}\sum_{d=1}^{\infty}
N^{\textrm{GW}}_{g,d}d^{2}q_{t}^{d},\, q_{t}=e^{t}\,.
\end{equation}
This implies in particular that
\begin{equation}\label{qtauqt}
q_{\tau}=\exp 2\pi i \tau=q_{t}(1+\mathcal{O}(q_{t}))\,.
\end{equation}
It is natural to expect that there should be an enumerative problem
associated to $\tau$ in the sense
\begin{equation}\label{Fgqtau}
F_{g}(\tau)=\sum_{d=0}^{\infty}N^{\textrm{hyp}}_{g,d}q_{\tau}^{d}\,,
\end{equation}
where like the Gromov-Witten invariants $N^{\textrm{GW}}_{g,d}$, the
numbers $N^{\textrm{hyp}}_{g,d}$ may hypothetically counting certain
kind of invariants. Comparing the (\ref{Fgqtau}) with (\ref{Fgqt})
and using (\ref{qtauqt}), we can then find the ``multiple-cover
formula" relating $N^{\textrm{GW}}_{g,d}$ and
$N^{\textrm{hyp}}_{g,d}$.

For the cases \cite{Alim:2013} in which the topological string
partition functions have nice expressions in terms of quasi-modular
forms in $\tau$, the integrality with respect to $q_{\tau}$ is
almost automatic. One then hopes that according to (\ref{qtauqt}),
studying the enumerative meaning of $q_{\tau}$ expansion will help
understand the integrality in $q_{t}$ expansion as well.

We don't have answers to any of these questions, and shall only
display some examples below.

\begin{expl}[Resolved Conifold]
Consider the resolved conifold which is the total space of
$\mathcal{O}_{\mathbb{P}^{1}}(-1)\oplus
\mathcal{O}_{\mathbb{P}^{1}}(-1)\rightarrow \mathbb{P}^{1}$ and is a
Calabi-Yau 3--fold. The Picard-Fuchs equation of the mirror Calabi-Yau family is
given by, see e.g. \cite{Forbes:2005}
\begin{equation}
\mathcal{L}=\theta ({\alpha \over 1-\alpha})^{-1}\theta^{2}\,.
\end{equation}
Near the large complex structure limit given by $\alpha=0$, a basis
of the periods could be chosen to be
\begin{equation}
X^{0}=1,\quad t=\ln \alpha,\quad F_{t}\sim
(\ln\alpha)^{2}+\mathcal{O}(\alpha^{0})\,.
\end{equation}
Therefore near $\alpha=0$, one has $t=\ln \alpha$ and thus
$q_{t}=\alpha$. Moreover, the genus zero Gromov-Witten invariants
are \cite{Aspinwall:1993, Manin:1995, Voisin:1996}
\begin{equation}
N^{\textrm{GW}}_{0,d}={1\over d^{3}}\,,
\end{equation}
and the prepotential is
\begin{equation}
F(t)={\kappa\over
3!}t^{3}+\sum_{d=1}^{\infty}N^{\textrm{GW}}_{0,d}q_{t}^{d}={\kappa\over
3!}t^{3}+\sum_{d=1}^{\infty}{1\over d^{3}}q_{t}^{d}={\kappa\over
3!}t^{3}+\mathrm{Li}_{3}(q_{t})\,.
\end{equation}
This implies in particular that
\begin{equation}\label{ctttofconifold}
C_{ttt}=\kappa+\sum_{d=1}^{\infty}q_{t}^{d}=\kappa+{q_{t}\over
1-q_{t}}\,.
\end{equation}
The function $\tau$ then satisfies
\begin{equation}\label{tautconifold}
2\pi i\tau=\kappa^{-1}F_{tt}=
t+\kappa^{-1}\sum_{d=0}^{\infty}{1\over d}q_{t}^{d}=
t-\kappa^{-1}\ln (1-q_{t})\,.
\end{equation}
Note that $\kappa$ can not be determined by studying the periods and
is ambiguous. Consideration in physics \cite{Vafa:2001} tells that a
natural choice is $\kappa=1$. In the following, we shall take this
choice.
\begin{rem}
From (\ref{tautconifold}) one can see that $\tau$ is itself the
generating function of the sequence of numbers ${1\over
d}=d^{2}N^{\textrm{GW}}_{0,d},d=1,2,\cdots$. These numbers appear in
the study of the stable-quotient invariants defined in
\cite{Marian:2011} with
\begin{equation}
d^{2}N^{\textrm{GW}}_{0,d}=\int_{[Q_{0,2}(\mathbb{P}^{1},d)]^{vir}}e(Ob)\cup
ev^{*}_{1}H\cup ev^{*}_{2}H\,,
\end{equation}
where $Ob$ is the obstruction bundle in the construction of
stable-quotient invariants, and the two insertions which give rise
to $ev^{*}_{1}H\cup ev^{*}_{2}H$ are required for the stability in
genus $0$.
\end{rem}

For higher genus partition functions, it is well known that
\cite{Faber:1998}
\begin{eqnarray}
&N^{\textrm{GW}}_{g,d}&=d^{2g-3}N^{\textrm{GW}}_{g,1}={|B_{2g}|\over 2g(2g-2)!}d^{2g-3},\\
&F_{g}&={|B_{2g}|\over 2g(2g-2)!}\mathrm{Li}_{3-2g}(q_{t})\,\quad
\text{in particular,}\quad F_{1}=-{1\over 12}\log (1-q_{t})\,.
\end{eqnarray}
To extract the numbers $N^{\textrm{hyp}}_{g,d}$ associated to
$\tau$, we make use of (\ref{tautconifold}) which gives rise to
\begin{equation*}
2\pi i\tau=t-\ln (1-q_{t}),\quad q_{\tau}={q_{t}\over 1-q_{t}},\quad
q_{t}={q_{\tau}\over 1+q_{\tau}}\,.
\end{equation*}
Now from (\ref{ctttofconifold}), one gets
\begin{equation}
C_{ttt}=\kappa t+{q_{t}\over 1-q_{t}}=\kappa \ln {q_{\tau}\over
1+q_{\tau}}+q_{\tau}= 2\pi
i\kappa\tau+q_{\tau}-\kappa\sum_{k=1}^{\infty}(-1)^{k}q_{\tau}^{k}\,.
\end{equation}
It follows that
\begin{equation*}
N^{\textrm{hyp}}_{0,d}=\kappa,\,\quad d=0,\quad
N^{\textrm{hyp}}_{0,d}=1+\kappa,\,\quad d=1,\,\quad
N^{\textrm{hyp}}_{0,d}= (-1)^{d+1}\kappa,\,\quad d\geq 2\,.
\end{equation*}
 For the
generating function $\partial_{t}F_{1}$, we get
\begin{equation}
\partial_{t}F_{1}={1\over 12}{q_{t}\over 1-q_{t}}={1\over 12}q_{\tau}\,.
\end{equation}
This then tells that
\begin{equation*}
N^{\textrm{hyp}}_{1,d}={1\over 12},\, d=1,\quad
N^{\textrm{hyp}}_{1,d}=0,\,d\geq 2\,.
\end{equation*}
For higher genus partition functions, we have
\begin{equation}
\sum_{d=1}^{\infty}N^{\textrm{hyp}}_{g,d}q_{\tau}^{d}={|B_{2g}|\over
2g(2g-2)!}\mathrm{Li}_{3-2g}(q_{\tau}(1+q_{\tau})^{-1})={|B_{2g}|\over
2g(2g-2)!}\theta_{q_{t}}^{2g-3}q_{\tau}\,.
\end{equation}
Since $\theta_{q_{t}}:=q_{t}{\partial \over \partial
q_{t}}=(1+q_{\tau})\theta_{q_{\tau}}$, one can then find
$N^{\textrm{hyp}}_{g,d}$ by direct computations. For any $g\geq 2$,
the first few invariants with $d=1,2,3\cdots$ are listed as follows:
\begin{eqnarray*}
N^{\textrm{hyp}}_{g,d}: &&1,\, -2+4^{2-g}, \,6-3* 2^{5-2 g}+2* 9^{2-g}\,,\\
&&-24+3* 2^{9-4 g}-8*3^{5-2 g}+9* 4^{3-g},\\
&& 120 \left(1-2^{8-4 g}-2^{5-2 g}+5^{3-2 g}+2* 9^{2-g}\right)\cdots
\end{eqnarray*}
\end{expl}
\begin{expl} [Local $\mathbb{P}^{2}$] Now we consider the Calabi-Yau 3--fold $K_{\mathbb{P}^{2}}$. In \cite{Alim:2013},
the holomorphic limits of the first few topological string partition
functions are solved genus by genus in terms of quasi-modular forms
and have nice expansions in $q_{\tau}$. For example,
\begin{eqnarray*}
C_{ttt}&=&
-{1\over 3}{\eta(3\tau)^{3}\over \eta(\tau)^{9}}\\
&=&-{1\over 3}\left(1+9 q_{\tau}+54 q_{\tau}^2+252 q_{\tau}^3+1008 q_{\tau}^4+3591 q_{\tau}^{5}+\cdots\right)\,,\\
\partial_{t}F^{1}&=&-{1\over 12} DF^{1}\cdot \kappa^{-1}C_{ttt}
=-{1\over 12}{3E_{2}(3\tau)+E_{2}(\tau)\over 4}{\eta(3\tau)^{3}\over \eta(\tau)^{9}}\\
&=&-{1\over 48}\left(1+3 q_{\tau}-18 q_{\tau}^2-276 q_{\tau}^3-1896
q_{\tau}^4-9675 q_{\tau}^5+\cdots\right)\,.
\end{eqnarray*}
From these expansions one can immediately read off the numbers
$N^{\textrm{hyp}}_{g,d},g=0,1,\,d=1,2,\cdots$. The canonical
coordinate $t$ is the following function of $q_{\tau}$:
\begin{equation*}
{1\over 2\pi i}{\partial t\over
\partial\tau}=\kappa C_{ttt}^{-1}={ \eta(\tau)^{9}\over
\eta(3\tau)^{3}},\quad t=\int {dq_{\tau}\over q_{\tau}}~{
\eta(\tau)^{9}\over \eta(3\tau)^{3}}\,,
\end{equation*}
The constant from integration is fixed by comparing the asymptotic
behaviors of $t$ and $\tau$ as $\tau\rightarrow i\infty $. A numeric
experiment using Mathematica shows that
\begin{equation*}
q_{t}=q_{\tau}- 9 q_{\tau}^2 + 54 q_{\tau}^3 - 246 q_{\tau}^4 + 909
q_{\tau}^5 - 2808 q_{\tau}^6 + 7299 q_{\tau}^7 -
 15705 q_{\tau}^8 +\cdots
\end{equation*}
See \cite{Mohri:2000kf, Stienstra:2005wy} for more discussions on this.
\end{expl}

\bigskip{}


\providecommand{\href}[2]{#2}\begingroup\raggedright\endgroup

\end{nouppercase}
\end{document}